\documentclass[lettersize,journal,10pt]{IEEEtran}
\usepackage{algorithmic}
\usepackage{algorithm}
\usepackage[caption=false,font=normalsize,labelfont=sf,textfont=sf]{}
\usepackage{textcomp}
\usepackage{stfloats}
\usepackage{url}
\usepackage{verbatim}
\usepackage{graphicx}
\usepackage{titlesec}
\usepackage[]{todonotes}
\titlespacing{\subsection}{1pt}{3pt}{3pt}

\usepackage{booktabs}

\usepackage{amsthm,amssymb, amsmath}
\usepackage[english]{babel}
\usepackage{tabularx} 

\usepackage{array}
\usepackage{amsmath}
\usepackage{multicol}
\usepackage{float}  \usepackage{tikz}
\usepackage{amsfonts}
\usepackage{multirow}
\usepackage{caption}
\usepackage{hyperref}
\usepackage{subcaption}
\usepackage{colortbl}
\usepackage{tikz}
\usetikzlibrary{positioning, arrows.meta}

\makeatletter
\def\footnoterule{\kern-3\p@
\hrule \@width 2in \kern 2.6\p@} \makeatother

\newcommand{\locref}[2]{\hyperref[#1]{#2~\ref{#1}}}
    
\newcommand{\loceqref}[2]{\hyperref[#1]{#2~(\ref{#1})}}
\newcommand{\cqfd}{\mbox{}\nolinebreak\hfill\rule{2mm}{2mm}\medbreak\par}

\newtheorem{remark}{Remark}
\newtheorem{theorem}{Theorem}
\newtheorem{lemma}{Lemma}
\newtheorem{definition}{Definition}

\newtheorem{assumption}{Assumption}

\DeclareMathOperator*{\argmax}{arg\,max}
\DeclareMathOperator*{\argmin}{arg\,min}

\newcommand{\abs}[1]{\left|#1\right|}

\usepackage{url}

\usepackage{enumitem}

\begin{document}

\title{Game-Theoretic Modeling of Stealthy Intrusion Defense against MDP-Based Attackers}

\author{Willie KOUAM, Johannes Kepler University Linz, Austria, arnold.kouam\_kounchou@jku.at\\
Stefan RASS, Johannes Kepler University Linz, Austria, stefan.rass@jku.at}

\maketitle

\begin{abstract}

The rapid expansion of Internet use has increased system exposure to cyber threats, with advanced persistent threats (APTs) being especially challenging due to their stealth, prolonged duration, and multi-stage attacks targeting high-value assets. In this study, we model APT evolution as a strategic interaction between an attacker and a defender on an attack graph. With limited information about the attacker’s position and progress, the defender acts at random intervals by deploying intrusion detection sensors across the network. Once a compromise is detected, affected components are immediately secured through measures such as backdoor removal, patching, or system reconfiguration. Meanwhile, the attacker begins with reconnaissance and then proceeds through the network, exploiting vulnerabilities and installing backdoors to maintain persistent access and adaptive movement. Furthermore, the attacker may take several steps between consecutive defensive operations, resulting in an asymmetric temporal dynamic. The defender's goal is to reduce the likelihood that the attacker will gain access to a critical asset, whereas the attacker's purpose is to increase this likelihood. We investigate this interaction under three informational regimes, reflecting varying levels of attacker knowledge prior to action: (i) a Stackelberg scenario, in which the attacker has full knowledge of the defender’s strategy and can optimize accordingly; (ii) a blind regime, where the attacker has no information and assumes uniform beliefs about defensive deployments; and (iii) a belief-based framework, where the attacker holds accurate probabilistic beliefs about the defender’s actions. For each regime, we derive optimal defensive strategies by solving the corresponding optimization problems.

\end{abstract}


\section{Introduction}

The increasing frequency and impact of cyber attacks require proactive defense mechanisms that enable early detection and mitigation. In particular, APTs pose a significant challenge due to their stealthy nature and their ability to progress through multiple stages over extended periods of time. Their evolution is commonly described by the cyber kill chain \cite{kamhoua2018game}, which encompasses reconnaissance, exploitation, command and control, privilege escalation, lateral movement, and ultimately, the compromise of a specific target. This work proposes a game-theoretic framework to support the operational management of network defense against stealthy and sophisticated attackers. As observed in numerous practical cases, we assume that the attacker has already gained an initial foothold in the system, and the defender’s objective is to prevent further compromise. The interaction is therefore modeled as a strategic game between a defender protecting critical assets and an attacker attempting to reach them. The battlefield in which this confrontation takes place is represented by an attack graph, a standard tool for security assessment in enterprise networks, constructed from threat modeling and refined through penetration testing and vulnerability scanning (see, for instance \cite{barik2016attack}). Within this environment, we consider a repeated stealthy intrusion game: the attacker progresses along one or more attack paths toward a critical asset during periods when the defender is inactive, while the defender intervenes randomly to disrupt this progression. The game ends when the attacker reaches the critical asset, resulting in a loss for the defender.

This problem has been previously studied in the \emph{Cut-The-Rope (CTR)} game introduced by Rass et al. \cite{rass2019cut, rass2023game}. In that model, the attacker simultaneously advances along all available attack paths, analogous to moving multiple pieces on a chessboard. Since the defender cannot observe the attacker’s exact position on any path, they consider a cohort of attacker “clones” starting from every possible node and moving toward the target. From the defender’s perspective, the attacker’s strategic choice reduces to selecting an attack path $\pi$, even though the precise progress along that path remains unknown. The defender then acts as if all clones on $\pi$ advance simultaneously. A key advantage of the CTR framework is its flexibility, allowing defenders to act decisively even with limited information. 
However, the standard CTR framework is modeled as a normal-form game, where the attacker lacks insight into the defense strategy, ignoring to a certain extent the operational reality of APTs. In practice (especially when the network administrator or defender acts first), the reconnaissance phase is a cornerstone of modern cyber-attacks, providing attackers with some understanding of the network. To address this limitation, we extend the CTR framework by explicitly incorporating reconnaissance information into the attacker’s problem, modeling it as a Markov Decision Process (MDP). In this extended model, the defender acts first and allocates intrusion detectors across the network to detect ongoing attacks. We therefore analyze three scenarios reflecting different levels of attacker information:
\begin{enumerate}[label=\roman*.]
    \item Stackelberg game approach: In this configuration, we assume the adversary possesses perfect knowledge of the defender's strategy. This leader-follower configuration evaluates the defender’s performance in a worst-case environment, where the attacker can optimally respond to committed defenses. \label{first_scenario_test}

    \item Probabilistic inference: The attacker possesses partial information and bases decisions on a probability distribution of expected defensive actions, using prior beliefs to guide progression.

    \item Strategic blindness (reconnaissance disruption): The attacker has no reliable knowledge of the defender’s behavior, often due to proactive measures such as moving target defense or cyber deception, which effectively neutralize reconnaissance and force the adversary to operate blindly.
\end{enumerate}

This scenario~\eqref{first_scenario_test} naturally gives rise to a Stochastic Stackelberg Game (SSG), where the defender, acting as the leader, commits to a detection strategy under resource constraints, and the attacker, as the follower, computes a best-response policy within the induced MDP. Similar formulations have been studied in the literature. For example, Li et al. \cite{li2022synthesis} proposed a mixed-integer linear program for jointly allocating detectors and stealthy sensors, while Sengupta et al. \cite{sengupta2018moving} formulated the issue of deploying a limited number of IDS in a large cloud network as a Stackelberg Game between a stealthy attacker and the cloud administrator. Other studies, like \cite{ma2023optimizing}, model uncertainty through a finite set of attacker types, each associated with a distinct reward function, and propose robust defense strategies based on worst-case absolute regret minimization. Resource allocation problems of this nature have also been widely studied in other contexts, such as preventing pollution attacks \cite{Tong2016Unified} or minimizing a joint cost function that accounts for the weighted trade-off between energy expenditure and execution latency, encompassing both MEC-based computational tasks and blockchain-specific processing requirements \cite{Chen2025Blockchain}, or preventing epidemic spread in networks \cite{Zhao2019Propagation, kouam2025exploring}. Our approach exploits the formalism in \cite{ma2023optimizing}, where the defender's actions directly parameterize the transition and reward structures of the attacker's MDP. By modeling the defense as a structural shift in the adversary’s decision environment, we establish our framework based on the following foundational assumptions.

\begin{itemize}
    \item The attack graph $G = (V, E)$ is a directed acyclic graph, where nodes represent system vulnerabilities or compromised states, and edges correspond to exploits or transitions between components. All attack paths lead to some  designated target asset $v_t \in V$. The acyclic property reflects that an attacker would not intentionally loop infinitely, if an exploit was successful. This, however, does not preclude the return to an earlier stage in the attack graph, for any reason (e.g., dead end discovery or defender's actions).

    \item The defender acts first by deploying automated updates or remediation actions whenever malicious activity is detected, while the attacker is already present at an unknown location in the system. However, due to operational constraints (e.g., lack of real-time alert transmission, periodic offline analysis, or one-time sensors that do not retain memory \cite{depolli2024offline}), the defender does not receive any alert information before making the subsequent deployment decision. Consequently, no belief updating can occur, and the defender acts under complete uncertainty at each stage. The defender, therefore, assumes that the attacker’s location is uniformly distributed across all possible nodes in the attack graph, except target nodes, to avoid trivial cases.

    \item Before executing the attack, the adversary performs a reconnaissance phase during which it gathers information about the network and the defender’s defensive actions. Based on this information, the attacker plans its strategy and may launch parallel or concurrent attacks, potentially exploiting multiple attack paths simultaneously to increase the likelihood of reaching the target.

    \item Recognizing that enterprise network defenders require clear, decisive actions to secure a network, we adopt a pure strategy approach. This ensures our findings align with real-world decision-making processes where randomized responses are often difficult to implement.
\end{itemize}

Following the CTR framework \cite{rass2023game}, we analyze a key operational regime that encodes the timing of actions for both attackers and defenders. This regime captures realistic operational constraints, producing dynamic interactions between the two parties. In our model, defender deployments occur at random intervals following an exponential distribution, reflecting environments with continuous monitoring, such as organizations with fully staffed security operations centers. Defensive actions thus follow a Poisson process with intensity $\lambda_D$, while offensive actions unfold within the stochastic pauses between deployments. Within these exponentially distributed intervals, the attacker executes a Poisson($\lambda$)-distributed number of steps, with the pause duration being governed by the defender's activity rate $\lambda_D$. The main contributions of this paper are as follows:

\begin{enumerate}
    \item We extend the CTR framework by modeling the attacker’s progression as a MDP in which state-dependent routing decisions at compromised nodes account for knowledge about the defender’s deployments. This approach modifies the MDP’s transition and reward structures to capture detection and mitigation, accurately representing asymmetric temporal dynamics while preserving the attack-graph structure.

    \item We formalize the interaction as a game under three attacker-information regimes: (i) a full-information Stackelberg setting, (ii) a belief-based regime with probabilistic attacker intelligence, and (iii) a blind regime induced by reconnaissance disruption. Each regime defines a distinct defender optimization problem and yields a specific attacker best-response policy.

    \item For each information regime, we employ approaches including linear programming, mixed-integer optimization, and Monte Carlo–based approximations. We validate our methods using detailed attack graphs of the Unguard network\footnote{https://github.com/dynatrace-oss/unguard} and two robotic systems: the Modular Articulated Robotic Arm (MARA) and the Mobile Industrial Robotics MiR100 platform. These graphs are based on comprehensive vulnerability assessments conducted by security researchers\footnote{https://news.aliasrobotics.com/the-week-of-mobile-industrial-robots-bugs/}\footnote{https://aliasrobotics.com/case-study-threat-model-mara.php, https://aliasrobotics.com/case-study-pentesting-MiR.php}. Our results show that optimized detector placement significantly outperforms baseline heuristics and quantify the defensive value of limiting attacker reconnaissance.
\end{enumerate}

\section{Motivation and Model overview}

\subsection{Notation and symbols}

In the following, sets are denoted by capital letters ($S, T$), vectors by bold characters ($\mathbf{v}, \mathbf{w}$), scalars by lowercase letters ($x, y$), random variable by $\widetilde{N}$ and agents' action sets by cursive letters ($\mathcal{A}, \mathcal{B}$). Table~\ref{symbols} summarizes the notation, providing each symbol alongside its corresponding definition.

\begin{table}[ht!]
\centering
\caption{Table of Symbols}
\label{symbols}
\begin{tabular}{ll}
\hline
\textbf{Symbol} & \textbf{Definition} \\
\hline
$G=(V,E)$ & Directed acyclic attack graph \\
$V$, $E$ & Set of nodes and directed edges respectively\\
$F$ & Set of target nodes (critical assets) \\
$S$, $\boldsymbol{s} \in S$ & State space of the attacker MDP, system's state\\
$\mathcal{A}_2(s)$ & Attacker actions available at state $s$ \\
$\mathcal{A}_1$ & Defender actions (detector deployment) space \\
$\pi$, $\pi^*$ & Attacker policy, optimal attacker policy \\
$P(s'|s,a)$ & Transition probability from $\boldsymbol{s}$ to $\boldsymbol{s'}$ under action $a$ \\
$P^{\boldsymbol{p}}$ & Transition function modified by the vector $\boldsymbol{p}$; see \eqref{real_transition_under_x}\\
$\boldsymbol{x}(v)$ & Protection status of $v$, $1$ if protected and $0$ otherwise\\
$R(s,a)$ & Immediate reward value \\
$V_2^\pi(s)$ & Attacker value function under policy $\pi$ \\
$\Delta(D)$ & Probability simplex over $D$ \\
$\lambda$, $\lambda_D$ & Attacker infection rate, defender inspection rate \\
$\Pr(A)$ & probability that an event $A$ occurs\\
$T$ & Game termination time \\
\hline
\end{tabular}
\end{table}

\subsection{Motivation}

In the original CTR game, an APT infiltrates a corporate network at an unknown entry point and commits to a single attack path toward a critical asset $v_t$. During the defender's idle period, the attacker progresses stealthily by installing backdoors at each compromised node, advancing by a random number of steps governed by a Poisson distribution with rate $\lambda$. When the defender returns and performs a spot-check, if they inspect a node along the attacker's chosen path where a backdoor has been installed, the backdoor is discovered and removed, effectively ``cutting the rope'', forcing the attacker back to the node immediately preceding the detection point. The formulation requires a simultaneous play and captures a key asymmetry: the attacker must pre-commit to an entire path before the game begins, while the defender must reason about all possible paths simultaneously, without knowledge of the attacker’s entry point or progress. The attacker selects a fixed sequence of nodes $\nu \to v_1 \to v_2 \to \cdots \to v_t$ and must follow this exact path regardless of circumstances.

However, in a scenario where the defender plays first, this rigid pre-commitment fails to capture the tactical flexibility that sophisticated adversaries actually possess during intrusion campaigns. In this case, before acting, the attacker performs reconnaissance to gather information on the network and the defender’s deployments. This helps attackers make routing decisions dynamically as they progress through a network, evaluating their options at each compromised node based on their current position, resources expended, and remaining distance to the target. The single-path commitment model ignores this adaptive decision-making, treating the attacker as if they must blindly follow a predetermined route. We thus extend the CTR framework by modeling the attacker's decision-making as a MDP, where the attacker makes state-dependent routing decisions at each compromised node during the defender’s idle period, rather than committing to a complete path in advance. This adjustment reflects the attacker’s adaptive response to the defender’s initial actions, highlighting the tactical flexibility inherent to advanced intrusion campaigns.

\subsection{Model Overview}

We study a two-player zero-sum security game played on an attack graph $G = (V, E)$ which encodes the interdependencies between system vulnerabilities, with all attack paths eventually leading to a designated target node $v_t$. We assume that the attacker has already infiltrated the network before the game begins, and the defender acts first but does not know the attacker’s exact location. To account for this uncertainty, the defender adopts a conservative view and behaves as if the attacker could be present at any node in $V \setminus F$. Operationally, this is captured by placing a clone at every possible attacker location, with each clone advancing independently along potential attack paths. This abstraction naturally captures lateral movement: switching paths corresponds to advancing a different clone.

\noindent \emph{ - Temporal pattern:} The game proceeds in rounds, which depend on the defender’s activity schedule. We focus on a pattern where the defender acts at random times, with exponentially distributed idle periods; each round corresponds to one such idle interval. This scenario, referred to as an \emph{exponential defense strategy} \cite{van2013flipit}, abstracts away from the execution time of individual actions and assumes that any action can be completed within one unit of time.

\noindent \emph{ - Defender and attacker actions:} The defender selects nodes on which to deploy protection mechanisms. These protections automatically activate upon detecting malicious activity and may include malware removal, disabling vulnerable services, or invalidating attacker knowledge through configuration changes. During the defender’s idle period, the attacker acts: each clone advances along its chosen attack path by several steps $n$, drawn from a fixed distribution $f_{\widetilde{N}}$. The attacker’s effectiveness depends critically on its knowledge of the defender’s most recent actions, motivating the definition of distinct information scenarios following the attacker’s reconnaissance phase.

\noindent\emph{ - Information Scenarios:} Before advancing, the attacker performs a reconnaissance phase to infer the defender’s most recent actions. In realistic attacks, the attacker's knowledge varies with reconnaissance capabilities, intelligence gathering, and the defender’s operational security. A sophisticated APT with extensive pre-attack reconnaissance may estimate defensive deployments accurately, while opportunistic attackers operate with minimal intelligence. This asymmetry affects both attacker routing decisions and defender optimization. We analyze three information settings, illustrated in Figure~\ref{context_CTR} and formalized in Table~\ref{formalized_scenarios}:

\begin{enumerate}
    \item \textbf{Perfect information (Stackelberg-type game):} The attacker learns the defender’s protection exactly after reconnaissance and can optimally plan their path. This scenario could represent a pessimistic case where the defender's operational security has failed, potentially arising from insider threats or information leakage.

    \item \textbf{Imperfect information:} The attacker holds probabilistic belief $\boldsymbol{p}$ about the defender’s action (and bases their decisions on it), while the defender is uncertain about which belief the attacker holds. This setting mirrors real-world operations with partial situational awareness, including cases of insider threats or adversaries with incomplete reconnaissance. Defensive systems must remain effective under such epistemic uncertainty.

    \item \textbf{No information (Blind):} The attacker gains no reliable intelligence about the defender’s action and plans under maximum uncertainty. This could represent the optimistic case where defensive operations are fully opaque.
\end{enumerate}

\vspace{-0.3cm}
\begin{figure}[ht!]
    \centering
    \includegraphics[width=1\linewidth]{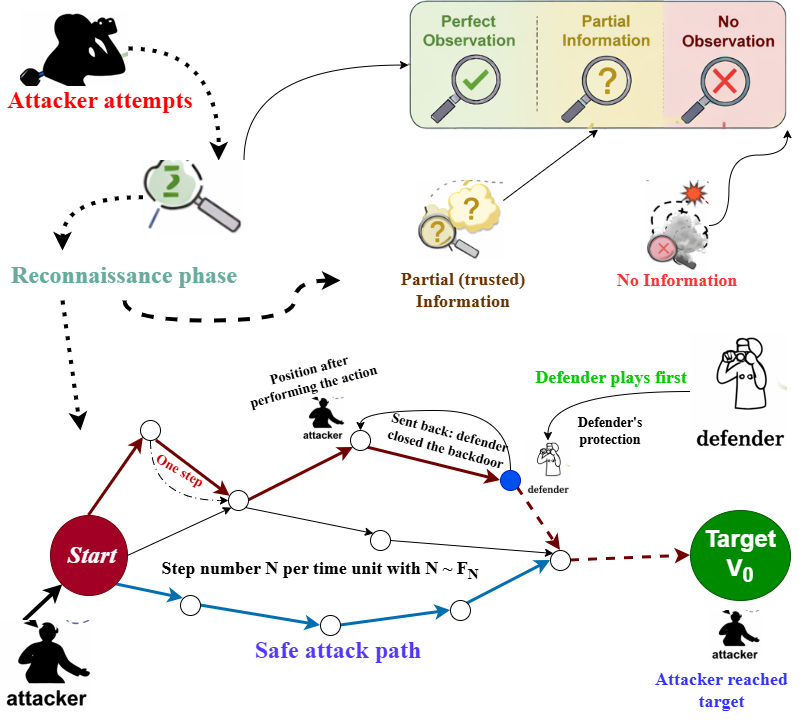}
    \caption{Illustration of the temporal and informational structure of the game}
    \label{context_CTR}
\end{figure}
\vspace{-0.3cm}

\begin{table}[h!]
\centering
\caption{Game-theoretic scenarios studied}
\label{formalized_scenarios}
\renewcommand{\arraystretch}{1.2}
\begin{tabular}{p{1.25cm} p{2.6cm} p{3.8cm}}
\hline
\textbf{Regime} & \textbf{Information structure} & \textbf{Objective and relevance} \\
\hline

\textbf{Blind} 
& No observation of opponent’s action (simultaneous move). 
& Baseline setting with hidden monitoring. Defender chooses a deployment strategy while the attacker solves MDP without observing it. This represents best-case defensive scenario. \\

\hline

\textbf{Stackelberg} 
& Attacker observes the defender’s deployment before acting. 
& Leader–follower model with attacker best-responding to observed defense. This captures worst-case visibility and quantifies defender’s disadvantage under exposure. \\

\hline

\textbf{Dirichlet} 
& Defender cannot completely conceal their actions, and the attacker learns deployment via observed frequencies. 
& Belief-based model where the defender exploits
observability to shape attacker’s belief using a Dirichlet prior, mitigating visibility loss. Results show that properly and strategically managed Dirichlet signaling reduces the Stackelberg disadvantage. \\

\hline
\end{tabular}
\end{table}

\noindent Each scenario induces a different strategic interaction, and for each, we seek the strategies maximizing each player's payoff. 

\noindent \emph{ - Strategies and payoffs:}
 
\begin{itemize}
    \item The defender selects $h$ nodes to protect within the set $A_1:= V \setminus F$ with $F$ the set of targets. When operational constraints limit deployable locations, we can restrict $A_1$ to a subset \(V_{\text{spot}} \subseteq V \setminus F\) (e.g., if the defender does not have access to several nodes in the network). The defender’s action space is therefore  expressed as $ \mathcal{A}_1 = \mathcal{P}_{h}(A_1)$ where $\mathcal{P}_{h}(A_1) $ denotes the $h$-element subsets of $A_1$. That means $\boldsymbol{x} \in \mathcal{A}_1$ denotes a protection strategy consisting of $h$ elements, with $h$ representing the defender’s resource budget. To reflect practical enterprise settings, we restrict attention to pure strategies, enabling clear and implementable decisions. At each active period, the defender deterministically selects a deployment vector \(\boldsymbol{x} \in \mathcal{A}_1\) formally the defined as $$\mathcal{A}_1 = \left\{\boldsymbol{x} \in \{0, 1\}^{|A_1|}, \sum_{\substack{j \in A_1}} x_j = h\right\}$$
    Each pure strategy \(\boldsymbol{x} \in \{0,1\}^{|A_1|}\) is the standard basis vector with $x_j = 1$ if $j$ is protected and $0$ elsewhere.

\item The attacker’s action space $\mathcal{A}_2$ consists of all attack paths in $G$ leading to a target. Each clone starts from a distinct node in $V \setminus F$ and follows one of these paths toward the target. The attacker may adopt a mixed strategy $\textbf{y} \in \Delta(\mathcal{A}_2)$, representing a probability distribution over attack paths. The payoff of the game is the probability that at least one avatar reaches a target $v_t$ during a single defender's idle period. This probability depends critically on the distribution of $N$, the number of steps the attacker can take while the defender is inactive. The attacker seeks to maximize this probability, while the defender seeks to minimize it. Formally,
$$u_{\text{attacker}} = -u_{\text{defender}} = \Pr(\text{attacker reaches } v_t).$$

\end{itemize}

\subsection{Network model and attack structure}

For each scenario that we evaluate in the following, the attacker plays on an oriented graph $G = (V, E)$ with the elements defined as follows,
\begin{itemize}
    \item $V = \{v_1, \ldots, v_n\}$ is the set of nodes or vulnerabilities
    \item $E \subseteq V \times V$ is the set of edges, i.e., exploitation transitions
    \item $F =\{v_t^{1}, \dots, v_t^{n}\}\subset V$ is the set of attacker's targets
    \item $V_{\text{entry}} \subseteq V \setminus F$ is the set of attacker entry points
    \item $\boldsymbol{s}_{\text{sink}}$ is the absorbing state when a detection occurs, while  $s_T$ is a terminal state when the attacker successfully reaches a target,  
    \item $c_t \leqslant c_{\max} = \max d(i,j)$ is the number of edges crossed by the attacker in the graph $G$, at a given time $t$.
\end{itemize}

\begin{definition}[Attack MDP structure] The attack process phase is modeled as a MDP $\textbf{M}_A = (S, A, P, \nu, F, R)$ where:
\begin{itemize}
\item $S = \left(\underbrace{V}_{\text{current position}}  \times \underbrace{\{0, 1, 2, \ldots, c_{\max}\}}_{\text{steps taken}}\right) \cup \quad \{\boldsymbol{s}_{\text{sink}}\}$ is the set of possible states of the system. In other words, a state in our framework is represented by a tuple containing the position and the number of edges traversed to reach that position. That is, a given absorbing state can be represented as $\boldsymbol{s}_{\text{sink}} = (v_i, c_i)$ (i.e., $\boldsymbol{s}_{\text{sink}}$ represents terminal state in which the attacker is detected) while a target state is defined as $\boldsymbol{s_T} = (v_T, c_T)$ (corresponding to a terminal state in which a target is reached).

\item $A(\boldsymbol{s}) = \{a_e: e = (p, u) \in E\} \ with\ \boldsymbol{s} \in S \setminus \{\boldsymbol{s}_{\text{sink}}, \boldsymbol{s_T}\}$ is the set of possible actions of the attacker at each state. Here $a_e$ means ``attempt exploitation along edge $e$", and $A(\boldsymbol{s}_{\text{sink}}) = A(\boldsymbol{s_T}) = \{\text{}\}$. Let's say that there is no further action after reaching an absorbing state (or a terminal state), the step game is then over. 

\item \(P: S \times A \to \Delta(S)\) is the transition function. For a given action \(a = a_e\) where \(e = (p,u)\), i.e., $\boldsymbol{s} = (p, c_t)$, the transition probability is defined as:

{\footnotesize{\begin{align}\label{Main_transition_function}
    P(\boldsymbol{s'} \mid \boldsymbol{s}, a_e) =
\begin{cases}
\Pr\bigl(\widetilde{N} \geqslant c_t + 1 \mid \widetilde{N} \geqslant c_t \bigr), \ \text{if } \boldsymbol{s'} = (u, c_t + 1)& \\[2pt]
1 - \Pr\bigl(\widetilde{N} \geqslant c_t + 1 \mid \widetilde{N} \geqslant c_t \bigr),\  \text{if }\ \boldsymbol{s'} = (p, c_t),& \\[2pt]
0, \qquad \qquad \qquad \qquad \qquad \qquad \qquad \ \text{otherwise,}&
\end{cases}
\end{align}}
}\nobreak
\noindent where \(\widetilde{N} \in \{0, 1, 2, \cdots\}\) is the random number of steps undertaken by the clone with probability mass function \(f_{\widetilde{N}}(n) = \Pr(\widetilde{N} = n)\), and \(c_t\) denotes the number of edges crossed by the attacker until the current position.

\noindent In other words, the attacker moves from node \(p\) (state \(\boldsymbol{s}\)) to node \(u\) (state \(\boldsymbol{s'}\)) with probability equal to the conditional probability that the number of steps available in this time interval exceeds the distance already covered by at least one step, i.e., $\Pr\bigl(\widetilde{N} \geqslant c_t + 1 \mid \widetilde{N} \geqslant c_t \bigr) = \dfrac{\sum_{n = c_t + 1}^{\infty} f_{\widetilde{N}}(n)}{\sum_{n = c_t}^{\infty} f_{\widetilde{N}}(n)}.$ Otherwise, the attacker remains in \(p\).

\item $\boldsymbol{\nu}=(\nu_{\boldsymbol{s}})_{\boldsymbol{s}\in S}$ is the initial state distribution of the attacker.

\item $R: S \times A \to \mathbb{R}$ is the attacker's reward function. Since the attacker's objective is to maximize the probability of reaching the target we define,
\begin{align*}
R(s, a) = \begin{cases}
1 & \text{if } s \in F \times \{c_t\}_{t \in \mathbb{N}} \text{ (target reached)} \\
0 & \text{otherwise}
\end{cases}
\end{align*}
\end{itemize}
\end{definition}

\noindent Note that the graph does not contain any cycles in our framework, and therefore the game always ends. Under this consideration, the value function $V^{\pi}_2(\boldsymbol{s_0})$ represents the probability that the attacker starting from state $\boldsymbol{s_0} = (v, 0)$ will eventually reach a target $v_t$ under policy $\pi$, directly capturing the attacker's objective: \textit{maximizing his chance of reaching a critical asset}.

\noindent\textbf{Attacker's objective formalization:} The attacker's value function under policy $\pi$ is defined, for any state $s = (v_i, c_i) \in S$, as
$V^{\pi}_2(\boldsymbol{s}) = \mathrm{E}_{\pi}\!\left[\displaystyle\sum_{k=0}^{\infty} R(\boldsymbol{s_k}, a_k) \,\Big|\, \boldsymbol{s_0} = \boldsymbol{s} \right]$.
The optimal value function satisfies the Bellman equation:
\begin{align*}
V^*_2(\boldsymbol{s}) 
= \max_{a \in \mathcal{A}_2(\boldsymbol{s})} 
\left(
R(\boldsymbol{s},a) + \sum_{\boldsymbol{s'} \in S} P(\boldsymbol{s'} \mid \boldsymbol{s},a)\, V^*_2(\boldsymbol{s'})
\right),
\end{align*}
and the optimal policy is given by
\begin{align*}
\pi^*(\boldsymbol{s}) 
\in &\argmax_{a \in \mathcal{A}_2(\boldsymbol{s})}
\left(
R(\boldsymbol{s},a) + \sum_{\boldsymbol{s'} \in S} P(\boldsymbol{s'} \mid \boldsymbol{s},a)\, V^*_2(\boldsymbol{s'})
\right),\\
& \forall \boldsymbol{s} = (v_i, c_i) \in S.
\end{align*}

\begin{lemma}
In our model, assuming that the attacker follows a policy $\pi$ generating the path \ $ \pi:\; p_0 = v_0 \to p_1 \to p_2 \to \cdots \to p_{L_\pi} = v_t$, of length $L_{\pi}$ (number of edges from $v_0$ to $v_t$ when following $\pi$), the value of the MDP at the initial state $\boldsymbol{s_0} = (v, 0)$ equals the probability that the random horizon $\widetilde{N}$ admits at least $L_\pi$ steps, that is, $V_2^{\pi}(\boldsymbol{s_0}) \;=\; \Pr(\widetilde{N} \geqslant L_{\pi})$.

\end{lemma}

\begin{proof}
For $c\in \mathbb{N}$ we consider $f_{\widetilde{N}}(c)\;=\;\Pr(\widetilde{N} \geqslant c) = \displaystyle\sum_{k=c}^{\infty}\Pr(\widetilde{N} = k)$ and then $f_{\widetilde{N}}(0) = 1$. For $c\geqslant 0$ with $f_{\widetilde{N}}(c) >0$ define the conditional step evolution as,
\[
\alpha_{c_{t+1}}:=\Pr(\widetilde{N} \geqslant c_{t} + 1 \mid \widetilde{N} \geqslant c_{t})=\dfrac{f_{\widetilde{N}}(c_{t} + 1)}{f_{\widetilde{N}}(c_t)}.
\]\noindent
In our framework, a transition attempt from state $(p, c_t)$ along the chosen edge $(p \to u)$ succeeds (i.e., the attacker actually advances to $u$) with probability $\alpha_{c_{t+1}}$ and fails (due to time expiration) with probability $1-\alpha_{c_{t+1}}$. We interpret the failure outcome as producing no further chance to reach $v_t$ in this round, matching the intended post-window semantics: if the window ends before the $c_{t+1}^{th}$ step, the attacker cannot reach $v_t$ in the round.\\  

\noindent \emph{Base step:} Consider the last move on the path, from $p_{L_\pi-1}$ to $p_{L_\pi}=v_t$.
At state $p_{L_\pi-1}$ the next-step conditional probability is $\alpha_{L_\pi}$. If the step is allowed, the attacker reaches $v_t$ and obtains reward $1$; otherwise, the round ends with a reward of $0$, i.e.,
$V_2^{\pi}\big(p_{L_\pi-1}, L_\pi-1\big) \;=\; \alpha_{L_\pi}\cdot 1 + (1-\alpha_{L_\pi})\cdot 0
\;=\; \alpha_{L_\pi}$.\\

\noindent \emph{Inductive step:}
Fix an index $j$ with $0\leqslant j\leqslant L_\pi-1$. From state $\boldsymbol{s_j} =  (p_j, j)$ the policy $\pi$ defines the unique next-state $\boldsymbol{s_{j + 1}} =(p_{j+1}, j+1)$. The Bellman equation gives

{\footnotesize{\begin{align*}
    V_2^{\pi}\big(\boldsymbol{s_{j}}\big)
\;=\; \alpha_{j+1}\cdot V_2^{\pi}\big(\boldsymbol{s_{j + 1}}\big) + (1-\alpha_{j+1})\cdot 0 \;=\; \alpha_{j+1}\,V_2^{\pi}\big(\boldsymbol{s_{j + 1}}\big).
\end{align*}}}

\noindent Applying this relation repeatedly from $j=L_\pi-1$ back to $j=0$ yields
\begin{align*}
    V_2^{\pi}\big(\boldsymbol{s_{0}}\big) = \prod_{k=1}^{L_\pi} \alpha_k = \prod_{k=1}^{L_\pi} \dfrac{f_{\widetilde{N}}(k)}{f_{\widetilde{N}}(k-1)} = \dfrac{f_{\widetilde{N}}(L_\pi)}{f_{\widetilde{N}}(0)} \;=\; f_{\widetilde{N}}(L_\pi)
\end{align*}
By definition $f_{\widetilde{N}}(L_\pi)=\Pr(\widetilde{N}\geqslant L_\pi)$, which proves the lemma.
\end{proof}

\section{Perfect Information (Stackelberg Game)}

\subsection{Problem formulation}

In the Stackelberg setting, the defender first commits to a strategy $\boldsymbol{x} \in \mathcal{A}_1$, the attacker then \emph{observes} $\boldsymbol{x}$ with certainty, and computes an optimal policy $\pi^{\text{Stack}}(\boldsymbol{x}) $ for the modified MDP $\textbf{M}_A(\boldsymbol{x}) $; the attack then proceeds according to $\pi^{\text{Stack}}(\boldsymbol{x}) $. Given the attacker's perfect knowledge of $\boldsymbol{x}$, the transition function is modified to:

{\small{\begin{align}\label{real_transition_under_x}
   P^{\boldsymbol{x}}(\boldsymbol{s'} \mid \boldsymbol{s} = (v,c), a) = \begin{cases}
1 & \text{if } \boldsymbol{x}(v) = 1 \text{ and }  \boldsymbol{s'} =  \boldsymbol{s}_{\text{sink}} \\
P(\boldsymbol{s'} \mid \boldsymbol{s}, a) & \text{otherwise}
\end{cases} 
\end{align}}}

\noindent That is, the attacker knows with certainty protected nodes, and routes accordingly. Thus, if the attacker passes through a protected node, detection occurs with certainty, and the process transitions immediately to the sink state. Otherwise, transitions follow the base dynamics $P(\boldsymbol{s'} \mid \boldsymbol{s}, a)$. Under these transition dynamics:
\begin{itemize}
    \item The sink state $\boldsymbol{s}_{\text{sink}}$ is an \emph{absorbing state} and represents successful detection, and the game ends (for the attacker).
    \item Upon detection, the attacker is reset to the node immediately preceding the protected node.
\end{itemize}

The modified transition function $P^{\boldsymbol{x}}$captures how the defender’s placement decision alters the attacker’s decision environment (i.e., success probability). Actually,
\begin{itemize}
    \item \emph{The base transition $P(\boldsymbol{s'} \mid \boldsymbol{s}, a)$}: models physical and temporal feasibility, i.e., whether sufficient steps remain to move from state $\boldsymbol{s}$ to $\boldsymbol{s'}$ via action $a$, independently of defensive deployments.
    
    \item \emph{The modified transition $P^{\boldsymbol{x}}(\boldsymbol{s'} \mid \boldsymbol{s} = (v, c_t), a)$}: incorporates strategic risk. If the current node is protected  (i.e., $x(v) = 1$), occupying it results in certain detection, regardless of movement feasibility.
\end{itemize}
\noindent Hence, the defender’s action introduces a structural constraint into the attacker’s MDP by transforming specific states into detection-triggering states. 
The value function $V^{\pi}_2(\boldsymbol{s_0} = (v_0, 0)) $ computed under $P^{\boldsymbol{x}}$ represents the probability that an attacker starting from entry point $v_0$ and following policy $\pi$ reaches a target $v_t$ without detection. This is precisely the quantity minimized by the defender and maximized by the attacker, leading to our equilibrium characterization. An equilibrium $(\boldsymbol{x^*}, \pi^*)$ represents a strategic configuration satisfying:
\begin{itemize}
    \item $\boldsymbol{x^*}$ minimizes the attacker’s optimal success probability; minimax strategy against the worst-case (most capable) attacker policy
    \item $\pi^*$ is a best-response to $\boldsymbol{x^*}$, optimizing their routing policy given anticipation of the defender's minimax choice
    \item Neither player can improve their payoff through unilateral deviation.
\end{itemize}
\noindent 
This Stackelberg formulation reflects the operational setting in which the defender commits first, and the attacker adapts optimally after observing the deployment.

\subsection{Stackelberg Equilibrium}

In our zero-sum Stackelberg setting, the defender's optimal strategy $x$ can be obtained by solving the following optimization problem:
\begin{align}
    \boldsymbol{x^*} \in \underset{\boldsymbol{y} \in \mathcal{A}_1}{\argmin} \max_{\pi} V_{2}^{\pi}(\boldsymbol{\nu}, \boldsymbol{y}); \  V_{2}^{\pi}(\boldsymbol{\nu}, \boldsymbol{y}) = \displaystyle \sum_{\boldsymbol{s} \in S} \nu_{\boldsymbol{s}} V_{2}^{\pi}(s, \boldsymbol{y})
    \label{MILP_problem}
\end{align}
which in fact is equivalent to the following bi-level optimization problem,
\begin{align*}
    \displaystyle\min_{\boldsymbol{y} \in \mathcal{A}_1}\quad & V_2^{\pi^*}(\boldsymbol{\nu}, \boldsymbol{y}) \\
   s.t.\quad  & \pi^{*} \in \argmax_{\pi \in \Pi} V_2^{\pi}(\boldsymbol{\nu}, \boldsymbol{y})
\end{align*}\noindent
where $\boldsymbol{\nu} = \left(\nu_{\boldsymbol{s}}\right)_{\boldsymbol{s} \in S}$ is selected to be the initial distribution over the states $S$, $V_{2}^{\pi}(\mathbf{s}; \boldsymbol{y})$ is attacker's success probability under policy $\pi$ in the modified MDP $\textbf{M}_A(\boldsymbol{x})$. To solve the attacker's MDP, let's consider a decision vector $\boldsymbol{v} = (v_{\boldsymbol{s}})_{\boldsymbol{s} \in S}$, where $v_{\boldsymbol{s}}$ is an upper bound on $v_{\boldsymbol{s}}^\ast$ for each $\boldsymbol{s} \in S$, with $\boldsymbol{v^*} = \left(v_{\boldsymbol{s}}^*\right)_{\boldsymbol{s} \in S}$ is the optimal value vector describing the probability of reaching target from state $\boldsymbol{s}$ under the optimal attack policy. The attacker’s value can then be computed by the following linear program:
\begin{align}
\min_{\boldsymbol{v}} \quad & \sum_{\boldsymbol{s} \in S} \nu_{\boldsymbol{s}} v_{\boldsymbol{s}} \label{eq:lp_obj} \\
\text{s.t.} \quad 
& v_{\boldsymbol{s}} \geqslant \sum_{s' \in S} P^{\boldsymbol{x}}(\boldsymbol{s'} \mid \boldsymbol{s}, a)\, v_{\boldsymbol{s'}}, 
&& \forall a \in \mathcal{A}_2,\ \forall \boldsymbol{s} \in S, \label{eq:lp_bellman}\\
& v_{\boldsymbol{s}} = 0, 
&& \forall s \in \{\boldsymbol{s}_{\text{sink}}\}, \label{eq:lp_sink}\\
& v_{\boldsymbol{s}} = 1, 
&& \forall \boldsymbol{s} \in F, \label{eq:lp_terminal}\\
& v_{\boldsymbol{s}} \geqslant 0, 
&& \forall \boldsymbol{s} \in \boldsymbol{S}. \label{eq:lp_nonneg}
\end{align}\noindent
  
\noindent It is shown in \cite{de2003linear} that any vector $\boldsymbol{v}$ satisfying ~\eqref{eq:lp_bellman} is an upper bound on the optimal value  $\boldsymbol{v^\ast}$, and equality is achieved by the optimal attack policy $\pi^*$ via the Bellman equation. The problem~\eqref{MILP_problem} is known to be a strong NP-hard problem \cite{hansen1992new}; we thus formulate the optimal deployment problem as a Mixed-Integer Linear Programming (MILP) problem. For this purpose, observe that the transition function $P^{\boldsymbol{x}}$ of the modified MDP $\textbf{M}_A(\boldsymbol{x})$ leads to the following equality 
\begin{align*}
    \sum_{\boldsymbol{s'}} & P^{\boldsymbol{x}}(\boldsymbol{s'} \mid \boldsymbol{s} = (v, c_t) ,a) v_{\boldsymbol{s'}} = \sum_{\boldsymbol{s'}} P(\boldsymbol{s'} \mid \boldsymbol{s} ,a)\, v_{\boldsymbol{s'}} (1 - \boldsymbol{x}(v) ) \\
    & + \boldsymbol{x}(v)  V_2^{\pi}(\boldsymbol{s}_{\text{sink}}, \boldsymbol{x})  = \sum_{\boldsymbol{s'}} P(\boldsymbol{s'}\mid \boldsymbol{s} ,a)\, v_{\boldsymbol{s'}} (1 - \boldsymbol{x}(v) )
\end{align*}\noindent
since we know that $V_2^{\pi}(\boldsymbol{s}_{\text{sink}}, \boldsymbol{x})  = 0.$ The optimization problem~\eqref{MILP_problem} can therefore be expressed as,
\begin{align}
&\min_{\boldsymbol{x}\in \mathcal{A}_1, \boldsymbol{v}} \qquad
     \sum_{\boldsymbol{s} \in S} \nu_{\boldsymbol{s}}\, v_{\boldsymbol{s}}  \label{eq:attacker-mdp-obj} \qquad \text{s.t.} \\[0.4em]
& 
 v_{\boldsymbol{s}} \geqslant \sum_{\boldsymbol{s'}\in S} 
      P\left(\boldsymbol{s'}\mid \boldsymbol{s} ,a \right)\, v_{\boldsymbol{s'}}\,\bigl(1 - \boldsymbol{x}(v) \bigr)
      \quad \forall a\in \mathcal{A}_2, \forall \boldsymbol{s} \in S,
      \label{eq:attacker-mdp-bellman}\\[-0.3em]
      & \sum_{u \in V_{\text{spot}}} \boldsymbol{x}(u) \leqslant h,\quad \boldsymbol{x}(u) \in \{0,1\} \qquad \quad \forall u \in V_{\text{spot}},  \label{dirichlet_single}\\ 
    & \text{and with constraints \eqref{eq:lp_sink}, \eqref{eq:lp_terminal}, and \eqref{eq:lp_nonneg}.} \notag
\end{align}

\noindent The Bellman constraint~\eqref{eq:attacker-mdp-bellman} contains a non-linear term $v_{\boldsymbol{s'}}\,\bigl(1 - \boldsymbol{x}(v) \bigr)$. Introducing auxiliary variables 
\begin{equation}
 W_{2,\boldsymbol{s}, \boldsymbol{s'}} =  v_{\boldsymbol{s'}}\,\bigl(1 - \boldsymbol{x}(v) \bigr) = 
\begin{cases}
    v_{\boldsymbol{s'}} & \text{if } \boldsymbol{x}(v) = 0; \\[0.3em]
    0 & \text{if } \boldsymbol{x}(v) = 1,
\end{cases}
\label{eq:w-definition}
\end{equation}
the constraint \eqref{eq:attacker-mdp-bellman} is linearized and equivalent to \begin{align}
    v_{\boldsymbol{s}} \geqslant \displaystyle\sum_{\boldsymbol{s'} \in S} P(\boldsymbol{s'} \mid \boldsymbol{s}, a)\, W_{2,\boldsymbol{s},\boldsymbol{s'}},\ \forall \boldsymbol{s} \in S,\; \forall a \in \mathcal{A}_2
    \label{eq:milp-equivalent}
\end{align}.

\noindent Using a big-$M$ method~\cite{griva2009linear}, the definition \eqref{eq:w-definition} can be expressed through the following inequalities:
\begin{subequations}
\label{eq:bigM}
\begin{align}
    W_{2,\boldsymbol{s},\boldsymbol{s'}}  \leqslant M \, (1 - \boldsymbol{x}(v) ) \quad &
    W_{2,\boldsymbol{s},\boldsymbol{s'}} \geqslant m \, (1 - \boldsymbol{x}(v)) \label{eq:bigM-b}\\
    W_{2,\boldsymbol{s},\boldsymbol{s'}} - v_{\boldsymbol{s'}}  \leqslant M\, \boldsymbol{x}(v) \quad &
    W_{2,\boldsymbol{s},\boldsymbol{s'}} - v_{\boldsymbol{s'}}  \geqslant m\, \boldsymbol{x}(v) \label{eq:bigM-d}
\end{align}
\end{subequations}
where $M$ and $m$ are constants chosen sufficiently large and small, respectively. $M = 1$ and $m = -1$ is sufficient in our setting, since when $x(v) =1$, constraints \eqref{eq:bigM-b} enforce $W_{2,\boldsymbol{s},\boldsymbol{s'}} = 0$, while \eqref{eq:bigM-d} become non-binding; and similarly for the case $x(v) = 0$. Finally, the MILP formulation for the defender’s Stackelberg problem becomes:
\begin{align}
\min_{\boldsymbol{x} \in \mathcal{A}_1, \boldsymbol{v}} \ & \sum_{\boldsymbol{s} \in S} \nu_{\boldsymbol{s}}\, v_{\boldsymbol{s}} \label{eq:milp-objective}\\[0.3em]
\text{s.t.} \qquad 
    &\text{constraints } \eqref{eq:lp_sink}, \eqref{eq:lp_terminal}, \eqref{eq:lp_nonneg}, \eqref{dirichlet_single}, \eqref{eq:milp-equivalent}, \eqref{eq:bigM} \label{eq:milp-constraints}\\[0.4em]
    & W_{2,\boldsymbol{s},\boldsymbol{s'}} \geqslant 0, 
    \qquad \forall \boldsymbol{s} \in S,\ \forall a \in \mathcal{A}_2,\ \forall \boldsymbol{s'} \in S \label{eq:milp-nonneg}
\end{align}

\section{Blind Attacker with No Information} \label{blind_model}

In the no-information setting, the attacker has no knowledge about the defender's strategy and must plan under maximum uncertainty. The defender commits a deployment $\boldsymbol{x} \in \mathcal{A}_1$ without revealing it, while the attacker adopts a \emph{uniform belief} over all possible defender actions and computes their routing policy accordingly. The attack proceeds, and detection occurs if the realized path intersects the actual deployed $\boldsymbol{x}$.

\subsection{Uniform Belief Model}

Assuming no information, the attacker adopts a uniform belief over defender actions, $\boldsymbol{p^{\text{Blind}}}(\boldsymbol{x})  = \dfrac{1}{|\mathcal{A}_1|}, \quad \forall \boldsymbol{x} \in \mathcal{A}_1.$ Under this belief, the attacker perceives transition dynamics from a given state $\boldsymbol{s} = (v,c)$ as,

{\small{\[
\boldsymbol{p^{\text{Blind}}}(\boldsymbol{s'} \mid \boldsymbol{s}, a) = \begin{cases}
\boldsymbol{p^{\text{Blind}}}(v), & \text{if } \boldsymbol{s'} = \boldsymbol{s}_{\text{sink}} \\
\bigl(1 - \boldsymbol{p^{\text{Blind}}}(v)\bigr) \cdot P(\boldsymbol{s'} \mid \boldsymbol{s}, a),  & \text{if } \boldsymbol{s'} \neq \boldsymbol{s}_{\text{sink}}
\end{cases}
\]}}
\noindent where ${\small{\boldsymbol{p^{\text{Blind}}}(v) = \dfrac{\displaystyle\sum_{\substack{\boldsymbol{x} \in \mathcal{A}_1}} \boldsymbol{x}(v)}{{\mid \mathcal{A}_1\mid}} = \dfrac{h}{|V_{\text{spot}}|} = \dfrac{\dbinom{|V_{\text{spot}}| - 1}{N - 1}}{\dbinom{|V_{\text{spot}}|}{N}}}}$, since once the node to protect is chosen, just search among the accessible nodes for the other ones to protect. At each node $v$, the attacker believes there is a $\boldsymbol{p^{\text{Blind}}}(v)$ probability of detection. The attacker plans their route to maximize expected success under this probabilistic detection model, rather than deterministic detection as in the Stackelberg case, and the uninformed attacker equilibrium is defined as follows.

\begin{definition}
\label{def:Blind}
A uninformed or blind attacker equilibrium is a strategy profile $(\boldsymbol{x^{\text{N}}}, \pi^{\text{N}})$ satisfying:

\noindent \emph{(i) Attacker's best-response under uniform belief:} 
\[
\pi^{\text{N}} \in \argmax_\pi V_2^\pi(\boldsymbol{\nu}; \boldsymbol{p^{\text{Blind}}}),
\]
where $V_2^\pi$ is computed on the belief-modified MDP with transitions $\boldsymbol{p^{\text{Blind}}}$.

\emph{(ii) Defender's best-response:} Given the attacker's policy $\pi^{\text{N}}$, the defender evaluates the \emph{actual} attacker success probability under deterministic detection:
\[
\boldsymbol{x^{\text{N}}} \in \argmin_{\boldsymbol{x} \in \mathcal{A}_1} V_2^{\pi^{\text{N}}}(\boldsymbol{\nu}; \boldsymbol{x}),
\]
where $V_2^{\pi^{\text{N}}}(\boldsymbol{\nu}; \boldsymbol{x}) $ is the attacker's success probability when following fixed policy $\pi^{\text{N}}$ against the defender's deterministic deployment $\boldsymbol{x}$.
\end{definition}

\subsection{Blind attacker equilibrium computation}

The blind attacker equilibrium is computed in two steps:

- \noindent \emph{Computation of the attacker's policy under uniform belief:} The attacker’s optimal value vector $\boldsymbol{v^{\text{Blind}}} = (v_{\boldsymbol{s}}^{\text{Blind}})_{\boldsymbol{s} \in S}$ is computed by solving the following LP:
\begin{subequations}
\label{eq:lp_Blind}
\begin{align}
&\min_{\boldsymbol{v}} \quad \sum_{\boldsymbol{s} \in S} \nu_{\boldsymbol{s}}\, v_{\boldsymbol{s}} \quad  \text{s.t.}  \label{eq:Blind_lp_obj} \\
& v_{\boldsymbol{s}} \geqslant\sum_{\boldsymbol{s'} \in S} \boldsymbol{p^{\text{Blind}}}(\boldsymbol{s'} \mid \boldsymbol{s}, a)\, v_{\boldsymbol{s'}},\ \forall \boldsymbol{s} \in S, \forall a \in \mathcal{A}_2, \label{eq:Blind_lp_bellman} \\
& v_{\boldsymbol{s}} = 0,\qquad \qquad \qquad \qquad \qquad \qquad  \forall \boldsymbol{s} \in \{\boldsymbol{s}_{\text{sink}}\}, \label{eq:Blind_lp_sink} \\
& v_{\boldsymbol{s}} = 1,\qquad \qquad \qquad \qquad \qquad \qquad  \forall \boldsymbol{s} \in F, \label{eq:Blind_lp_terminal} \\
& v_{\boldsymbol{s}} \geqslant 0,\qquad \qquad \qquad \qquad \qquad \qquad \qquad  \forall \boldsymbol{s} \in S. \label{eq:Blind_lp_nonneg}
\end{align}
\end{subequations}
\noindent where the transition function $\boldsymbol{p^{\text{Blind}}}$ is expanded as:
\begin{align*}
 &\boldsymbol{p^{\text{Blind}}}(\boldsymbol{s'} \mid \boldsymbol{s} = (v,c), a) 
= \bigl(1 - \boldsymbol{p^{\text{Blind}}}(v)\bigr) \cdot P(\boldsymbol{s'} \mid \boldsymbol{s}, a) \\
& +\ \boldsymbol{p^{\text{Blind}}}(v) \cdot \mathbb{I}_{s' = \boldsymbol{s}_{\text{sink}}} \\
&= \dfrac{h}{|V_{\text{spot}}|} \cdot \mathbb{I}_{s' = \boldsymbol{s}_{\text{sink}}} + \left(1 -\dfrac{h}{|V_{\text{spot}}|}\right) \cdot P(\boldsymbol{s'} \mid \boldsymbol{s}, a)
\end{align*}

\noindent The attack policy $\pi^{\text{N}}$ is recovered from the Bellman equation:
{\small{\[
\pi^{\text{N}}(s) \in \argmax_{a \in \mathcal{A}_1(s)} \left\{R(s, a) + \sum_{s'} \boldsymbol{p^{\text{Blind}}}(\boldsymbol{s'} \mid \boldsymbol{s}, a)\, v_{\boldsymbol{s'}}^{\text{Blind}}\right\}
\]}}

\noindent - \emph{Defender's optimal deployment strategy (via enumeration):} The defender's objective is the deployments strategy $\boldsymbol{x}$ that performs best against the ``blind-optimized'' attacker policy. That is, given the fixed policy $\pi^{\text{N}}$, the defender evaluates the \emph{actual} attacker success probability for each candidate $x \in \mathcal{A}_1$. For each $x \in \mathcal{A}_1$, the defender constructs the modified MDP $\textbf{M}_A(\boldsymbol{x}) $ with transition function $P^{\boldsymbol{x}}$ and evaluates the fixed policy $\pi^{\text{N}}$ in this MDP to obtain $V_2^{\pi^{\text{N}}}(\boldsymbol{\nu}, \boldsymbol{x})$ using,
\begin{align*}
& V_2^{\pi^{\text{N}}}(\boldsymbol{s}, \boldsymbol{x})  = \sum_{s'} P^{\boldsymbol{x}}(\boldsymbol{s'} \mid \boldsymbol{s}, \pi^{\text{N}}(\boldsymbol{s})) V_2^{\pi^{\text{N}}}(\boldsymbol{s'}; \boldsymbol{x})  \\
& \text{i.e., } V_2^{\pi^{\text{N}}}(\boldsymbol{\nu}, \boldsymbol{x}) = \sum_{\boldsymbol{s} \in S} \nu_{\boldsymbol{s}} V_2^{\pi^{\text{N}}}(\boldsymbol{s}, \boldsymbol{x}) 
\end{align*}

\noindent The defender then picks a minimizer $\boldsymbol{x^{\text{N}}} \in \displaystyle \argmin_{\boldsymbol{x} \in \mathcal{A}_1} V_2^{\pi^{\text{N}}}(\boldsymbol{\nu})$.

\noindent \emph{Computational complexity:}  Since $|\mathcal{A}_1| = \binom{|V_{\text{spot}}|}{N}$ grows combinatorially, enumeration is tractable for some $N$ and moderate $|V_{\text{spot}}|$. The complexity associated with this computation is then given by $O(|\mathcal{A}_1| \cdot |S|^2)$ while LP consumes $O(|\mathcal{A}_2| \cdot |S|^2)$ to solve~\eqref{eq:lp_Blind}.

\section{Belief-based defense under Dirichlet uncertainty}

In realistic attack scenarios, adversaries differ in intelligence capabilities. Some attackers operate with diffuse, low-quality intelligence resembling near-uniform beliefs, while sophisticated APTs possess concentrated beliefs formed through sustained reconnaissance. Consequently, the defender faces epistemic uncertainty regarding the attacker’s knowledge level. A robust defensive strategy must therefore perform well against a spectrum of attacker beliefs. In this setting, the attacker forms a belief about the defender’s deployment and optimizes accordingly, while the defender does not directly observe this belief. Importantly, when complete secrecy is infeasible, the defender is not obligated to accept the worst-case Stackelberg scenario or a scenario in which they are completely unaware of the attacker's methodology. Instead, the defender may strategically manipulate the attacker’s belief. By deliberately revealing controlled information, through observable system artifacts, configuration traces, or engineered operational security (OPSEC) leaks \cite{almeshekah2014planning, rowe2004two}, the defender can induce a desired belief structure and exploit the mismatch between perceived and actual deployments. The defender strategically induces a target concentration vector $\boldsymbol{\alpha^{\text{ind}}} = (\alpha^{\text{ind}}_1, \alpha^{\text{ind}}_2, \ldots, \alpha^{\text{ind}}_{|V_{\text{spot}}|}) \in \mathbb{R}_{>0}^{|V_{\text{spot}}|}$ where $\alpha^{\text{ind}}_i$ represents the apparent frequency with which node $i$ is observed to be protected during reconnaissance. The deception campaign is designed so that the attacker’s belief is distributed according to $\boldsymbol{q} \sim \text{Dir}(\boldsymbol{\alpha^{\text{ind}}})$, a Dirichlet distribution \cite{lin2016dirichlet}. Because the defender engineers the deception process, the parameter $\boldsymbol{\alpha^{\text{ind}}}$ is known to the defender.

\subsection{Dirichlet model for belief uncertainty}

The attacker holds a private belief $\boldsymbol{q} \in \Delta(V_{\text{spot}})$ where $\boldsymbol{q}(v)$ represents the perceived detection intensity at node $v$. 
The defender does not observe the realization of $\boldsymbol{q}$, but models this uncertainty using the Dirichlet prior $\boldsymbol{p} \sim \mathrm{Dir}(\boldsymbol{\alpha^{\text{ind}}})$. This constitutes a second-order belief: the defender reasons over a distribution of possible attacker beliefs. 

\begin{remark}
    Unlike the blind model (Section~\ref{blind_model}), where the attacker knows the total number $h$ of protections and the marginals satisfy $\displaystyle \sum_v p^{\text{Blind}}(v) = h$, in the Dirichlet model, the attacker may misperceive both the number and distribution of protections. Individual Dirichlet samples may be highly concentrated; for instance, a sample could place most or all weight on a single node. In such cases, the attacker behaves as if only that node is likely protected, even though in reality $h>1$ nodes are protected. This modeling choice captures misperception or incomplete intelligence and allows the defender to exploit epistemic uncertainty through robust or deception-based strategies.
\end{remark} 

\noindent \emph{Attacker optimization given belief $\boldsymbol{q}$:} For a fixed belief vector $\boldsymbol{q}$ the attacker solves: $\pi^*(\boldsymbol{q}) = \displaystyle\argmax_{\pi^{\boldsymbol{q}}} V_2^{\pi^{\boldsymbol{q}}}(\boldsymbol{\nu}),$ where the value function is computed in the belief-modified MDP $\textbf{M}_A^{\boldsymbol{q}} = (S, A, P^{\boldsymbol{q}}, \nu, R)$ with the transition function expressed as:
{\small{\[
P^{\boldsymbol{q}}(\boldsymbol{s'} \mid \boldsymbol{s}=(v,c), a) = 
\begin{cases}
\boldsymbol{q}(v), & \text{if } \boldsymbol{s'} = \boldsymbol{s}_{\text{sink}} \\
(1 - \boldsymbol{q}(v))\,P(\boldsymbol{s'} \mid \boldsymbol{s}, a), & \text{if } \boldsymbol{s'} \neq \boldsymbol{s}_{\text{sink}}
\end{cases}
\]}}
\noindent Thus, the attacker routes while accepting detection risk proportional to $\boldsymbol{q}(v)$. Let $\boldsymbol{v^{q}} = (v_{\boldsymbol{s}}^{\boldsymbol{q}})_{\boldsymbol{s} \in S}$ denote the attacker’s optimal value function under belief $\boldsymbol{q}$. It is computed by solving:
\begin{subequations}
\label{eq:attacker_lp}
\begin{align}
&\min_{\boldsymbol{v}} \quad \sum_{\boldsymbol{s} \in S} \nu_{\boldsymbol{s}}\, v_{\boldsymbol{s}}  \quad  \text{s.t.} \label{eq:belief_lp_obj} \\
& v_{\boldsymbol{s}} \geqslant \sum_{s' \in S} P^{\boldsymbol{q}}(\boldsymbol{s'} \mid \boldsymbol{s}, a)\, v_{\boldsymbol{s'}}, 
&& \forall \boldsymbol{s} \in S,\ \forall a \in \mathcal{A}_2, \label{eq:belief_lp_bellman} \\
& v_{\boldsymbol{s}} = 0, 
&& \forall \boldsymbol{s} \in \{\boldsymbol{s}_{\text{sink}}\}, \label{eq:belief_lp_sink} \\& v_{\boldsymbol{s}} = 1, 
&& \forall \boldsymbol{s} \in F, \label{eq:belief_lp_terminal} \\
& v_{\boldsymbol{s}} \geqslant 0, 
&& \forall \boldsymbol{s} \in S. \label{eq:belief_lp_nonneg}
\end{align}
\end{subequations}
\noindent The optimal policy $\pi^*(\boldsymbol{q})$ is recovered via the Bellman optimality condition.

\noindent \emph{Defender's robust optimization:} After computing $\pi^*(\boldsymbol{q})$, the attacker executes this policy in the true MDP governed by the defender’s deterministic deployment $\boldsymbol{x}$, with transition kernel $P^{\boldsymbol{x}}$. Let $V_2^{\pi^*(\boldsymbol{p})}(\boldsymbol{\nu}, \boldsymbol{x}) $ denote the realized attack success probability when the attacker optimizes under belief $\boldsymbol{p}$, but the system transitions follow $P^{\boldsymbol{x}}$. The defender chooses $\boldsymbol{x}$ to minimize the expected attacker success:
\begin{align}
\label{def_expectation}
\min_{\boldsymbol{x} \in \mathcal{A}_1}
\mathrm{E}_{\boldsymbol{p} \sim \mathrm{Dir}(\boldsymbol{\alpha^{\text{ind}}})}
\left[
V_2^{\pi^*(\boldsymbol{p})}(\boldsymbol{\nu}, \boldsymbol{x})
\right].
\end{align}

\noindent This objective captures robustness against the entire distribution of possible attacker beliefs induced by the deception campaign. The expectation in \eqref{def_expectation} is, however, analytically intractable. The difficulty arises because the optimal policy $\pi^*(\boldsymbol{p})$ depends on $\boldsymbol{p}$ in a complex, nonlinear way (via the Bellman optimality equations), and no closed-form expression exists for this expectation. We therefore approximate the expectation using \emph{Monte Carlo sampling} over belief realizations, replacing the integral over the Dirichlet distribution by an empirical average over sampled belief vectors.

\subsection*{Monte Carlo approximation}

The defender draws $\boldsymbol{p_1},\ldots, \boldsymbol{p_K} \sim \mathrm{Dir}(\boldsymbol{\alpha^{\text{ind}}})$, and approximates
$\mathrm{E}_{\boldsymbol{p} \sim \mathrm{Dir}(\boldsymbol{\alpha^{\text{ind}}})}[V_2^{\pi^*(\boldsymbol{p})}(\boldsymbol{\nu}, \boldsymbol{x})]
\approx \dfrac{1}{K} \displaystyle\sum_{k=1}^K V_2^{\pi^*(\boldsymbol{p_k})}(\boldsymbol{\nu}, \boldsymbol{x})$.

\noindent For each sample $\boldsymbol{p_k}$, the attacker’s optimal policy $\pi^*(\boldsymbol{p_k})$ is computed by solving \eqref{eq:attacker_lp}. For $k \in \{1,\ldots,K\}$, let $v_{\boldsymbol{s},k}$ denote the attacker’s realized value at state $\boldsymbol{s}$ when executing $\pi^*(\boldsymbol{p_k})$ under deployment $\boldsymbol{x}$ (with $h$ protections available, ).
These values satisfy the Bellman evaluation equations:
\begin{align*}
v_{\boldsymbol{s},k} &= R\bigl(\boldsymbol{s},\pi^*(\boldsymbol{p_k})(s)\bigr) + \sum_{\boldsymbol{s'}} P^{\boldsymbol{x}}\bigl(\boldsymbol{s'} \mid \boldsymbol{s},\pi^*(\boldsymbol{p_k})(\boldsymbol{s})\bigr)\,v_{\boldsymbol{s'},k} \\
\text{with}& \ \sum_{\boldsymbol{s'}} P^{\boldsymbol{x}}(\boldsymbol{s'} \mid \boldsymbol{s} = (v,c), a)\, v_{\boldsymbol{s'},k}  
= \boldsymbol{x}(v) \cdot \underbrace{v_{\boldsymbol{s}_{\text{sink}},k}}_{=0}\\
&+ \bigl(1 - \boldsymbol{x}(v)\bigr) \sum_{\boldsymbol{s'} \neq \boldsymbol{s}_{\text{sink}}} P(\boldsymbol{s'} \mid \boldsymbol{s}, a)\, v_{\boldsymbol{s'},k}  \\
&= \sum_{\boldsymbol{s'}} P(\boldsymbol{s'} \mid \boldsymbol{s}, a) \cdot \underbrace{\bigl[\bigl(1 - \boldsymbol{x}(v)\bigr) v_{\boldsymbol{s'},k} \bigr]}_{w_{\boldsymbol{s}, \boldsymbol{s'}, k}}.
\end{align*}

\noindent The Bellman equations contain bilinear terms of the form $(1-\boldsymbol{x}(v))v_{\boldsymbol{s'},k} $. To linearize these products, we introduce auxiliary variables $w_{\boldsymbol{s}, \boldsymbol{s'} ,k} = (1 - \boldsymbol{x}(v))v_{\boldsymbol{s'},k} $. The defender’s optimization problem can then be written as the MILP:
\begin{subequations}
\label{eq:dirichlet_milp}
\begin{align}
\min_{\boldsymbol{x}, \boldsymbol{v}, \boldsymbol{w}} \ & \dfrac{1}{K} \sum_{k=1}^K \sum_{\boldsymbol{s} \in S} c_{\boldsymbol{s}}\, v_{\boldsymbol{s},k} 
\label{eq:dirichlet_obj} \\[0.3em]
\text{s.t.} \ 
& v_{\boldsymbol{s},k} \geqslant \sum_{\boldsymbol{s'} \in S} P(\boldsymbol{s'} \mid \boldsymbol{s}, \pi^*(\boldsymbol{p_k})(\boldsymbol{s}))\, w_{\boldsymbol{s},\boldsymbol{s'},k}, 
 \forall k, \boldsymbol{s}    \label{eq:dirichlet_bellman} \\[0.3em]
& w_{(v,c),\boldsymbol{s'},k} \leqslant v_{\boldsymbol{s'},k} + \boldsymbol{x}(v) , \quad  \forall k, \boldsymbol{s} =(v,c), \boldsymbol{s'} \label{eq:dirichlet_mcc1} \\
& w_{(v,c),\boldsymbol{s'},k} \leqslant 1 - \boldsymbol{x}(v), \qquad \forall k, \boldsymbol{s} = (v, c), \boldsymbol{s'}   \label{eq:dirichlet_mcc2} \\
& w_{(v,c),\boldsymbol{s'},k} \geqslant v_{\boldsymbol{s'},k}, \qquad \qquad  \forall k, \boldsymbol{s} = (v, c), \boldsymbol{s'} \label{eq:dirichlet_mcc3} \\
& w_{(v,c),\boldsymbol{s'},k} \geqslant 0,  \qquad \qquad \quad \forall k, \boldsymbol{s} = (v,c), \boldsymbol{s'} \label{eq:dirichlet_mcc4} \\[0.3em]
& v_{\boldsymbol{s},k} = 0,  \qquad \qquad \qquad \qquad \forall k, \boldsymbol{s} = \boldsymbol{s}_{\text{sink}} \label{eq:dirichlet_sink} \\
& v_{\boldsymbol{s},k} = 1,  \qquad \qquad \qquad \qquad \forall k, \boldsymbol{s} \in F \label{eq:dirichlet_terminal} \\[0.3em]
& \sum_{u \in V_{\text{spot}}} \boldsymbol{x}(u) \leqslant h \label{eq:dirichlet_single} \\
& \boldsymbol{x}(u) \in \{0,1\}, \qquad \qquad \qquad \qquad \forall u \in V_{\text{spot}} \label{eq:dirichlet_binary}
\end{align}
\end{subequations}

\begin{lemma} Let's define the following functions
$\widehat{F}_K(\boldsymbol{x}) = \dfrac{1}{K}\displaystyle\sum_{k=1}^K V_2^{\pi^*(\boldsymbol{p_k})}(\boldsymbol{\nu}, \boldsymbol{x}) ,
\ F(\boldsymbol{x})  = \mathrm{E}_{\boldsymbol{p} \sim \mathrm{Dir}(\boldsymbol{\alpha^{\text{ind}}})} \bigl[ V_2^{\pi^*(\boldsymbol{p})}(\boldsymbol{\nu}, \boldsymbol{x})\bigr]$ representing the sample-average approximation and the true expectation respectively. Let $\boldsymbol{p_1},\ldots, \boldsymbol{p_K}$ be independent draws from the same Dirichlet distribution $\mathrm{Dir}(\boldsymbol{\alpha^{\text{ind}}})$.
Then for any $\varepsilon>0$,\
$\Pr\!\left( \bigl| \widehat{F}_K(\boldsymbol{x})  - F(\boldsymbol{x}) \bigr| > \varepsilon\right)\leqslant 2 \exp\!\left(-2K\varepsilon^2\right)$.
Consequently, to guarantee $\bigl| \widehat{F}_K(\boldsymbol{x})  - F(\boldsymbol{x})  \bigr| \leqslant \varepsilon$ with probability at least $1-\delta$, it suffices to take $K \geqslant \dfrac{1}{2\varepsilon^2}\ln\!\left(\dfrac{2}{\delta}\right)$.
\end{lemma}

\begin{proof}
For each sampled belief $\boldsymbol{p_k}$, let $Z_k := V_2^{\pi^*(\boldsymbol{p_k})}(\boldsymbol{\nu}, \boldsymbol{x})$. Because the evaluation of a fixed policy yields a probability of reaching the target, we have $0 \leqslant Z_k \leqslant 1$. Since the samples $\boldsymbol{p_k}$ are independent and identically distributed (i.i.d.), the random variables $Z_k$ are i.i.d. and bounded in $[0,1]$, with $\mathrm{E}[Z_k] = F(\boldsymbol{x})$ Hoeffding's inequality for the empirical mean
$\widehat{F}_K(\boldsymbol{x}) = \dfrac{1}{K}\displaystyle\sum_{k=1}^K Z_k$ gives $\Pr\!\left(\left| \widehat{F}_K(\boldsymbol{x}) - F(\boldsymbol{x}) \right| > \varepsilon \right) \leqslant 2\exp(-2K\varepsilon^2)$.
\end{proof}

\subsection{Dirichlet vs Stackelberg: advantages of a Dirichlet-robust approach}

Let $d = \abs{V_{\text{spot}}}$, and let the attacker’s subjective belief about node protection be modeled as a probability vector $\boldsymbol{p} = (p_v)_{v \in V_{\text{spot}}} \in \Delta^{d-1}:= \displaystyle\left\{\boldsymbol{p} \in \mathbb{R}_+^d: \sum_{v} p_v = 1 \right\}$. This vector represents the attacker’s perception of how protection is distributed across nodes, and is normalized to sum to $1$, independent of the actual number $h$ of protected nodes. Each realization of a Dirichlet random vector $\boldsymbol{p}\sim\mathrm{Dir}(\boldsymbol{\alpha^{\text{ind}}})$ may be concentrated or diffuse, reflecting differences in attacker intelligence quality.

Let $ L(\boldsymbol{x}, \pi):= V_2^{\pi}(\boldsymbol{\nu}, \boldsymbol{x})$ be the realized attacker success probability under policy $\pi$ when the defender deploys \(\boldsymbol{x} \in \mathcal{A}_1\). For a fixed belief \(\boldsymbol{p}\), the attacker’s optimal policy is then expressed as \(\pi^*(\boldsymbol{p}) \in \displaystyle \argmax_{\pi(\boldsymbol{p}) \in \Pi} L(\boldsymbol{x}, \pi(\boldsymbol{p}))\), and the Stackelberg defender problem is $\displaystyle\boldsymbol{x^{\mathrm{St}}} \in \argmin_{\boldsymbol{x}\in\mathcal A_1} L(\boldsymbol{x},\pi^*(\boldsymbol{p}))$. Since the MDP is finite with  bounded rewards, we have $0 \leqslant L(\boldsymbol{x}, \pi(\boldsymbol{p})) \leqslant 1, \forall \boldsymbol{x}\in\mathcal A_1.$
\bigskip 

\noindent The Stackelberg optimal value is defined as, $V^{\mathrm{St}}:= \displaystyle\min_{\boldsymbol{x} \in \mathcal{A}_1} \max_{\pi \in \Pi} L(\boldsymbol{x}, \pi)$, while the Dirichlet-robust strategy minimizes the expected success probability under the belief distribution:

{\small{$$\boldsymbol{x^{\mathrm{Dir}}} \in \argmin_{\boldsymbol{x} \in \mathcal{A}_1} \mathrm{E}_{\boldsymbol{p} \sim \mathrm{Dir}(\boldsymbol{\alpha^{\text{ind}}})} \left[ L(\boldsymbol{x}, \pi^*(\boldsymbol{p})) \right] := \argmin_{\boldsymbol{x} \in \mathcal{A}_1} J_{\boldsymbol{\alpha^{\text{ind}}}}(\boldsymbol{x})$$}}

\noindent For any realization $\boldsymbol{p} \sim \mathrm{Dir}(\boldsymbol{\alpha^{\text{ind}}})$, the attacker's success satisfies.
\begin{align}
    & L(\boldsymbol{x^{\mathrm{St}}}, \pi^*(\boldsymbol{p})) \leqslant \max_{\pi \in \Pi} L(\boldsymbol{x^{\mathrm{St}}}, \pi) = V^{\mathrm{St}} \notag \\
    & \implies \mathrm{E}_{\boldsymbol{p} \sim \mathrm{Dir}(\boldsymbol{\alpha^{\text{ind}}})} \left[ L(\boldsymbol{x^{\mathrm{St}}}, \pi^*(\boldsymbol{p})) \right] \leqslant V^{\mathrm{St}} \label{worst_stack}\\
    & \text{as\ $\boldsymbol{x^{\mathrm{Dir}}}$\ is\ a\ minimizer,\ we\ obtain} \notag\\
    & \mathrm{E}_{\boldsymbol{p} \sim \mathrm{Dir}(\boldsymbol{\alpha^{\text{ind}}})} \left[ L(\boldsymbol{x^{\mathrm{Dir}}}, \pi^*(\boldsymbol{p})) \right] \leqslant \mathrm{E}_{\boldsymbol{p} \sim \mathrm{Dir}(\boldsymbol{\alpha^{\text{ind}}})} \left[ L(\boldsymbol{x^{\mathrm{St}}}, \pi^*(\boldsymbol{p})) \right] \notag\\
    & \iff J_{\boldsymbol{\alpha^{\text{ind}}}}(\boldsymbol{x^{\mathrm{Dir}}}) \leqslant \mathrm{E}_{\boldsymbol{p} \sim \mathrm{Dir}(\boldsymbol{\alpha^{\text{ind}}})} \left[ L(\boldsymbol{x^{\mathrm{St}}}, \pi^*(\boldsymbol{p})) \right] \leqslant  V^{\mathrm{St}} \notag.
\end{align}

\noindent so the Dirichlet optimization never performs worse than the Stackelberg solution. Moreover, this inequality can be strict. In particular, when the attacker’s best response changes discontinuously with beliefs, a Dirichlet distribution concentrated around specific belief regions may yield a strictly lower expected attacker success than the Stackelberg strategy. The next assumption and theorem formalize this result.

\begin{assumption}[boundary separability]\label{asm:bounded-separability}
Let \(\boldsymbol{p^*} \in\Delta^{d-1}\) be a belief at which the attacker's best-response mapping may be discontinuous (i.e., \(\boldsymbol{p^*}\) may lie on a policy-switching boundary). Suppose there exist defender deployments \(\boldsymbol{x^{\mathrm{St}}}\) and \(\boldsymbol{x'}\in\mathcal A_1\) and a sequence of beliefs \(\boldsymbol{p^{(n)}}\to \boldsymbol{p^*}\) such that $\displaystyle \limsup_{n \to \infty} L(\boldsymbol{x^{\mathrm{St}}}, \boldsymbol{p^{(n)}})
> \liminf_{n \to \infty} L(\boldsymbol{x'}, \boldsymbol{p^{(n)}})$, i.e., there exist a constant $\Delta > 0$ for which, for infinitely many \(n\),
\begin{equation}\label{eqn:strict-difference}
L(\boldsymbol{x^{\mathrm{St}}}, \boldsymbol{p^{(n)}}) \;\geqslant\; L(\boldsymbol{x'}, \boldsymbol{p^{(n)}}) + \Delta.
\end{equation}
Hence, we can select a subsequence of \(\boldsymbol{p^{(n)}}\to \boldsymbol{p^*}\) so that \eqref{eqn:strict-difference} holds for all $n$. Let us hereafter take \(\boldsymbol{p^{(n)}}\to \boldsymbol{p^*}\) to be this subsequence (to ease our notation).
\end{assumption}

\begin{theorem}\label{lem:dirichlet-concentration}
    Under assumption \ref{asm:bounded-separability} there exists a sequence of Dirichlet parameters \(\boldsymbol{\boldsymbol{\boldsymbol{\alpha_n}}} := M_n \boldsymbol{p^{(n)}}\) with \(M_n \to \infty\) and \(\boldsymbol{p^{(n)}}\to \boldsymbol{p^*}\) (the same sequence as in assumption \ref{asm:bounded-separability}) such that for all sufficiently large \(n\), $\quad J_{\boldsymbol{\alpha_n}}(\boldsymbol{x'}) \;<\; J_{\boldsymbol{\alpha_n}}(\boldsymbol{x^{\mathrm{St}}}).$
\end{theorem}

\begin{proof}

A standard construction of the Dirichlet distribution is via normalized Gamma random variables \cite{ferguson1973bayesian}:
Let $Y_1, \dots, Y_d$ be independent random variables with $Y_i \sim \mathrm{Gamma}(\alpha_{n,i}, 1),\ \alpha_{n,i} = M_n\, p_i^{(n)},\ i=1,\dots,d,$, where $p_i^{(n)}$ is the deterministic center of the Dirichlet and $M_n>0$ is a scaling parameter controlling concentration. Define the Dirichlet random variables $p_i = \dfrac{Y_i}{\displaystyle\sum_{j=1}^d Y_j}$, then $\mathrm{E}[p_i] = p_i^{(n)},\ 
\mathrm{Var}[p_i] = \dfrac{p_i^{(n)} (1 - p_i^{(n)})}{M_n + 1}$, thus as $M_n \to \infty$, the variance decreases and $p_i$ concentrates around $p_i^{(n)}$. More rigorously, let $Z_i$ = $\dfrac{Y_i}{M_n}$ then, $\mathbb E\left[Z_i\right] = p^{(n)}_i,\
\mathrm{Var}\!\left(Z_i\right) = \dfrac{p^{(n)}_i}{M_n}$. By Chebyshev's inequality \cite{rivlin2020chebyshev}, for any \(\varepsilon>0\),
\[
\Pr\!\Big(\big|Y_i/M_n - p^{(n)}_i\big| \geqslant \varepsilon\Big) \leqslant\dfrac{\mathrm{Var}(Y_i/M_n)}{\varepsilon^2} = \dfrac{p^{(n)}_i}{M_n\varepsilon^2} \xrightarrow[M_n\to\infty]{} 0.
\]
and by Slutsky's theorem \cite{delbaen2006remark}, $p_i = \dfrac{Y_i/M_n}{\displaystyle \sum_j Y_j/M_n} \xrightarrow{} p^{(n)}_i \ \text{as } M_n \to \infty$, and then $\boldsymbol{p} \xrightarrow{} \boldsymbol{p}^{(n)}$. Therefore, for any neighborhood \(U_n := \{\boldsymbol{p} \in \Delta^{d-1} \,:\, \|\boldsymbol{p} - \boldsymbol{p^{(n)}}\|_\infty < \varepsilon \}\) of \(\boldsymbol{p^{(n)}}\), $\displaystyle \Pr_{\boldsymbol{p}\sim\mathrm{Dir}(M_n \boldsymbol{p^{n}})}[\boldsymbol{p}\in U_n] \xrightarrow[M_n\to\infty]{} 1$. 
Moreover, by assumption~\ref{asm:bounded-separability} there exists a sequence \(\boldsymbol{p^{(n)}} \xrightarrow{} \boldsymbol{p^*}\) and a positive constant \(\Delta>0\) such that, $L\big(\boldsymbol{x}^{\mathrm{St}}, \boldsymbol{p^{(n)}}\big) \;\geqslant\; L\big(\boldsymbol{x'}, \boldsymbol{p^{(n)}}\big) + 2\Delta$.
Using boundedness and continuity of the utility function \(L\) (although the value function may be discontinuous), pick a small neighborhood $U_n$ around $\boldsymbol{p}^{(n)}$ such that for all $\boldsymbol{x} \in \mathcal{A}_1$ and $\boldsymbol{q} \in U_n$, $\displaystyle |L(\boldsymbol{x}, \boldsymbol{q}) - L(\boldsymbol{x}, \boldsymbol{p^{(n)}})| < \dfrac{\Delta}{4}.$ In this case,
\begin{align*}
& J_{\alpha_n}(\boldsymbol{x})  = \mathrm{E}[L(\boldsymbol{x}, \pi(\boldsymbol{p}))\mathbf{1}_{\{\boldsymbol{p}\in U_n\}}] + \mathrm{E}[L(\boldsymbol{x}, \pi(\boldsymbol{p}))\mathbf{1}_{\{\boldsymbol{p} \notin U_n\}}]
\end{align*}
On the set \(\{\boldsymbol{p}\in U_n\}\) we have \(L(\boldsymbol{x},\pi(\boldsymbol{p})) = L(\boldsymbol{x}, \pi(\boldsymbol{p^{(n)}})) + \delta(\boldsymbol{x}, \boldsymbol{p})\) with \(|\delta(\boldsymbol{x}, \boldsymbol{p})| < \Delta/4\). 
Hence
\begin{align*}
    \mathrm{E}[L(\boldsymbol{x}, \pi(\boldsymbol{p}))\mathbf{1}_{\{\boldsymbol{p}\in U_n\}}] & = L(\boldsymbol{x}, \pi(\boldsymbol{p^{(n)}}))\Pr(\boldsymbol{p}\in U_n)\\
    & + \delta(\boldsymbol{x}, \boldsymbol{p}) \Pr(\boldsymbol{p}\in U_n)
\end{align*}
\noindent Using the fact that \(|L(\boldsymbol{x}, \pi)|\leqslant 1\), \ $\forall$ $\boldsymbol{x} \in \mathcal{A}_1$, we obtain,
\begin{align*}
& \big| J_{\alpha_n}(\boldsymbol{x})  - L(\boldsymbol{x}, \boldsymbol{p^{(n)}}) \big|
 = \big| \mathrm{E}[L(\boldsymbol{x},\pi(\boldsymbol{p})) \mathbf{1}_{\{\boldsymbol{p} \in U_n\}}] \\
& - L(\boldsymbol{x},\pi(p^{(n)})) + \mathrm{E}[L(\boldsymbol{x},\pi(\boldsymbol{p})) \mathbf{1}_{\{p \notin U_n\}}] \big|\\
& = \big| L(\boldsymbol{x}, \pi(\boldsymbol{p^{(n)}}))\Pr(\boldsymbol{p}\in U_n) + \delta(\boldsymbol{x}, \boldsymbol{p}) \Pr(\boldsymbol{p}\in U_n) \\
& - L(\boldsymbol{x},\pi(\boldsymbol{p^{(n)}})) + \mathrm{E}[L(\boldsymbol{x},\pi(\boldsymbol{p})) \mathbf{1}_{\{\boldsymbol{p} \notin U_n\}}] \big|\\
& =  \big| - L(\boldsymbol{x}, \pi(\boldsymbol{p^{(n)}}))\Pr(\boldsymbol{p}\notin U_n) + \delta(\boldsymbol{x}, \boldsymbol{p}) \Pr(\boldsymbol{p}\in U_n) \\
& + \mathrm{E}[L(\boldsymbol{x},\pi(\boldsymbol{p})) \mathbf{1}_{\{p \notin U_n\}}] \big|\\
&\leqslant \dfrac{\Delta}{4} \Pr(\boldsymbol{p}\in U_n) + 4\Pr(\boldsymbol{p}\notin U_n).
\end{align*}

\noindent Since $\displaystyle \Pr_{\boldsymbol{p}\sim\mathrm{Dir}(M_n \boldsymbol{p^{n}})}[\boldsymbol{p}\in U_n] \xrightarrow[M_n\to\infty]{} 1$, for \(M_n\) sufficiently large, \(\Pr(\boldsymbol{p}\notin U_n) < \dfrac{\Delta}{16}\). Then
\begin{align*}
\big| J_{\alpha_n}(\boldsymbol{x})  - L(\boldsymbol{x}, \boldsymbol{p^{(n)}}) \big| &< \dfrac{\Delta}{4} + 4\cdot\dfrac{\Delta}{16} = \dfrac{\Delta}{2}\\
\iff J_{\alpha_n}(\boldsymbol{x^{\mathrm{St}}})
&\geqslant L(\boldsymbol{x^{\mathrm{St}}}, \boldsymbol{p^{(n)}}) - \dfrac{\Delta}{2} \\
&\geqslant \big(L(\boldsymbol{x'}, \boldsymbol{p^{(n)}}) + 2\Delta\big) - \dfrac{\Delta}{2} \quad\text{by \eqref{eqn:strict-difference}}\\
&> L\boldsymbol{(x'}, \boldsymbol{p^{(n)}}) + \dfrac{\Delta}{2} \geqslant J_{\alpha_n}(\boldsymbol{x'}).
\end{align*}
Therefore, 
$J_{\boldsymbol{\alpha^{\text{ind}}}}(\boldsymbol{x^{\mathrm{Dir}}}) \leqslant J_{\boldsymbol{\alpha^{\text{ind}}}}(\boldsymbol{x'}) <J_{\boldsymbol{\alpha^{\text{ind}}}}(\boldsymbol{x^{\mathrm{St}}})$. By exploiting~\ref{worst_stack}, $J_{\boldsymbol{\alpha^{\text{ind}}}}(\boldsymbol{x^{\mathrm{Dir}}}) < J_{\boldsymbol{\alpha^{\text{ind}}}}(\boldsymbol{x^{\mathrm{St}}}) \leqslant V^{St}$ with $\boldsymbol{\alpha^{\text{ind}}} = \alpha_n$.
\end{proof}

Under the assumption~\ref{asm:bounded-separability} and considering that the attacker and the defender follow the same Dirichlet distribution $\text{Dir}(\boldsymbol{\alpha^{\text{ind}}})$, the theorem proves that the Dirichlet-optimal strategy achieves strictly lower expected real attacker success than Stackelberg optimization.

\subsection{Concrete illustration for a 6-node network} 

Consider the MDP model instantiated on the following Figure~\ref{transition-dynamics-values}:
\noindent Path B: $A \to B \to T$ (2 steps)\\
Path C: $A \to C \to T$ (2 steps)\\
Path D - E: $A \to D \to E \to T$ (4 steps)
\begin{figure}[hb!]
\vspace{-0.75cm}
\centering
\begin{tikzpicture}[
    auto,
    >=Stealth,
    node distance=2cm and 2cm,
    state/.style={circle, draw, thick, minimum size=9mm, inner sep=0pt, font=\footnotesize},
    lbl/.style={font=\tiny, align=center}
]

\node[state] (A) {$s_A$};
\node[above=1mm of A, font=\tiny] {Entry};

\node[state, above right=1cm and 2cm of A] (B) {$s_B$};
\node[above=1mm of B, font=\tiny] {B};

\node[state, below right=1cm and 2cm of A] (C) {$s_C$};
\node[below=1mm of C, font=\tiny] {C};

\node[state, right=4cm of A] (D) {$s_D$};
\node[above=1mm of D, font=\tiny] {D};

\node[state, below=1.3cm of D] (E) {$s_E$};
\node[below=1mm of E, font=\tiny] {E};

\node[state, right=1.65cm of D] (T) {$s_T$};
\node[above=1mm of T, font=\tiny] {Target};

\node[state, above = 2cm of T] (sink) {$s_s$};
\node[above=1mm of sink, font=\tiny] {Sink};

\draw[->, thick] (A) -- node[above, lbl] {B} (B);
\draw[->, thick] (A) -- node[below, lbl] {C} (C);
\draw[->, thick] (A) -- node[above, lbl] {D-E} (D);

\draw[->, red, thick] (B) to[out=60,in=180] node[above, lbl] {$p_B$} (sink);
\draw[->, blue, thick, dashed] (B) -- node[above, lbl] {$1-p_B$} (T);

\draw[->, red, thick] (C) to[out=100,in=180] node[left, lbl] {$p_C$} (sink);
\draw[->, blue, thick, dashed] (C) -- node[above, lbl] {$1-p_C$} (T);

\draw[->, red, thick] (D) to[out=30,in=210] node[right, lbl] {$p_D$} (sink);
\draw[->, blue!60!black, thick, dashed] (D) -- node[below, lbl] {$1-p_D$} (E);

\draw[->, red, thick] (E) to[out=0,in=270] node[right, lbl] {$p_E$} (sink);
\draw[->, blue!60!black, thick, dashed] (E) -- node[below, lbl] {$1-p_E$} (T);

\draw[->, green!60!black, thick] (T) to[out=70,in=300] node[above, lbl] {success} (sink);

\node[below=3mm of E.south, lbl, align=center] {
  Short: $P=0.80$ \quad Long: $P=0.35$  \qquad $\mathrm{Val}_B = (1-p_B) \times 0.80$\\
   $\mathrm{Val}_C = (1-p_C) \times 0.80$ \qquad $\mathrm{Val}_{DE} = (1-p_D)(1-p_E) \times 0.35$
};
\end{tikzpicture}
\caption{Transition dynamics. Short paths success: 0.80; long path success: 0.35.}
\label{transition-dynamics-values}
\end{figure}
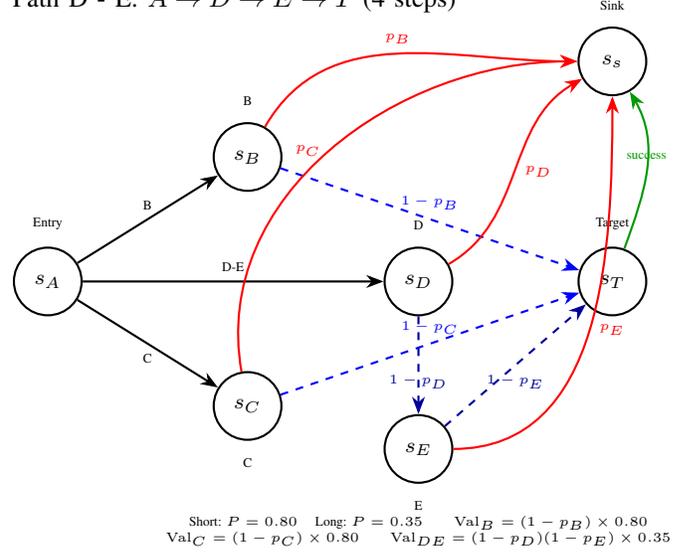

\noindent Defender can protect exactly one node from $\{B, C, D, E\}$) and the attacker success probabilities are given by $\Pr(N \geqslant 2) \approx 0.80$ (reach short paths) and $\Pr(N \geqslant 4) \approx 0.35$ (reach long path).
\noindent Given belief $p = (p_B, p_C, p_D, p_E)$ over protection probabilities, expected values under belief $p$ are expressed as:
\begin{align*}
\text{Value}(\text{path via } B \mid p) &= (1-p_B) \cdot 0.80 \\
\text{Value}(\text{path via } C \mid p) &= (1-p_C) \cdot 0.80 \\
\text{Value}(\text{path via } D,E \mid p) &= (1-p_D)(1-p_E) \cdot 0.35
\end{align*}

\noindent The Stackelberg game framework analysis of the game is,
\begin{table}[hb!]
\begin{tabular}{lllc}
\hline
Def. protects & Att. knows this & Att. chooses & Success prob. \\
\hline
$x = B$ & Belief: $(1,0,0,0)$ & Path $C$ & $0.80$ \\
$x = C$ & Belief: $(0,1,0,0)$ & Path $B$ & $0.80$ \\
$x = D$ & Belief: $(0,0,1,0)$ & Path $B$ or $C$ & $0.80$ \\
$x = E$ & Belief: $(0,0,0,1)$ & Path $B$ or $C$ & $0.80$ \\
\hline
\end{tabular}
\end{table}

\subsubsection*{ - Identify the critical point \(\boldsymbol{p}^*\) and the sequence \(\boldsymbol{p}^{(n)}\)}
From the analysis, the attacker switches policy when $p_B = p_C$. Let $\boldsymbol{p^*} = (0.25, 0.25, 0.25, 0.25)$, to satisfy assumption~\ref{asm:bounded-separability}.
\begin{align*}
&\text{Value}(\text{path B} \mid \boldsymbol{p^*}) = \left(1-\tfrac{1}{4}\right) \cdot 0.80 = \tfrac{3}{4} \cdot 0.80 = 0.60 \\
&\text{Value}(\text{path C} \mid \boldsymbol{p^*}) = \left(1-\tfrac{1}{4}\right) \cdot 0.80 = 0.60 \\
&\text{Value}(\text{path D,E} \mid \boldsymbol{p^*}) = \left(1-\tfrac{1}{4}\right)\left(1-\tfrac{1}{4}\right) \cdot 0.35  \approx 0.20
\end{align*}
\noindent Attacker's optimal policy under $p^*$ is then $\pi^*(p^*) = \text{Path B}$. Therefore,
$\text{Stackelberg avg} = 0.80$ and $\text{Dirichlet avg} = 0.60$.

\noindent Construct a sequence approaching $\boldsymbol{p}^*$ from the ``B-side'': $\boldsymbol{p}^{(n)} = \Big(0.25 + \frac{1}{n},\ 0.25 - \frac{1}{n},\ 0.25,\ 0.25\Big), \quad n\geqslant 5$; for all $n$, $p_B^{(n)} > p_C^{(n)}$, hence $\pi^*(\boldsymbol{p}^{(n)}) = \text{Path B}$.
Following Theorem~\ref{lem:dirichlet-concentration}, we construct a Dirichlet: $\boldsymbol{\alpha_n} = M_n \boldsymbol{p}^{(n)}, \quad M_n \to \infty$.
Let $\boldsymbol{p} \sim \mathrm{Dir}(\boldsymbol{\alpha_n})$, then, as $M_n \to \infty$, $\Pr\big(\|\boldsymbol{p} - \boldsymbol{p}^{(n)}\|_\infty < \varepsilon\big) \to 1$. 

\subsubsection*{ - Comparison of Stackelberg and Dirichlet-robust strategies along the sequence}
By considering:
\begin{itemize}
    \item $\boldsymbol x^{\mathrm{St}}$: any Stackelberg strategy (e.g., protect $C$, $D$, or $E$)
    \item $x' = \text{protect B}$
\end{itemize}
For any $\boldsymbol p^{(n)}$ in the sequence, the attacker follows Path B. Therefore, $L(\boldsymbol x^{\mathrm{St}}, \boldsymbol p^{(n)}) = 0.80 \text{ and } L(\boldsymbol x', \boldsymbol p^{(n)}) = 0$. Hence, $L(\boldsymbol x^{\mathrm{St}}, \boldsymbol p^{(n)}) = 0.80 \geqslant 0 + 0.80 = L(\boldsymbol x', \boldsymbol p^{(n)}) + \Delta$ with $\Delta = 0.80$. If $\boldsymbol x^{\mathrm{St}} = B$, we instead take $\boldsymbol x' = C$ and approach $\boldsymbol p^*$ from the C-side, obtaining the same gap. That is, assumption~\ref{asm:bounded-separability} holds with:  $\boldsymbol p^* = (0.25, 0.25, 0.25, 0.25)$, the sequence $\boldsymbol p^{(n)} \to \boldsymbol p^*$ from the B-side, $\boldsymbol x^{\mathrm{St}}$ (any Stackelberg strategy $\neq B$) and $\boldsymbol x' = B$, with a constant gap $\Delta = 0.80$ along the whole sequence. This is summarized by the following Table,
\begin{table}[ht!]
\centering
\begin{tabular}{lcccc|c}
\hline
Defender strategy & B & C & D & E & AVG. \\
\hline
Stackelberg & 0.80 & 0.80 & 0.80 & 0.80 & 0.80 \\
Dirichlet $x'$ & 0.00 & 0.80 & 0.80 & 0.80 & 0.60 \\
Improvement & 100\% & 0\% & 0\% & 0\% & 25\% \\
\hline
\end{tabular}
\caption{Expected attacker success (loss) for Stackelberg vs Dirichlet-robust strategy.}
\end{table}

The regime guiding the game dynamics is described in the following section.

\section{Movement pattern studied: defender with random intervals}

We consider a defender who performs deployments at random times according to a Poisson process with rate $\lambda_D > 0$. Consequently, the time interval during which the defender is absent is a random variable $T_D \sim \mathrm{Exp}(\lambda_D), \quad f_{T_D}(t) = \lambda_D e^{-\lambda_D t},\ t \geqslant 0$. The attacker is assumed to be active only while the defender is absent. During a defender's absence period of length $t$, the attacker performs penetration steps according to a Poisson process with rate $\lambda > 0$. Conditional on $T_D = t$, the number of attack steps $\widetilde{N}$ therefore follows a Poisson distribution: $\Pr[\widetilde{N} = n \mid T_D = t] = \dfrac{(\lambda t)^n}{n!} e^{-\lambda t},\ n = 0,1,2,\dots$. Since the duration $T_D$ of the defender's absence is itself random, the unconditional distribution of $\widetilde{N}$ is obtained by mixing the Poisson distribution with the exponential distribution:
\begin{align*}
&\Pr[\widetilde{N} = n]
=
\int_0^\infty
\Pr[\widetilde{N} = n \mid T_D = t]\,
f_{T_D}(t)\, dt\\
&=
\int_0^\infty
\dfrac{(\lambda t)^n}{n!} e^{-\lambda t}
\cdot
\lambda_D e^{-\lambda_D t}
\, dt = 
\dfrac{\lambda_D \lambda^n}{n!}
\int_0^\infty
t^n e^{-(\lambda + \lambda_D)t}
\, dt.
\end{align*}

\noindent The above is a standard Gamma integral. For any $a>0$ and integer $n \geqslant 0$,
\begin{align*}
    \int_0^\infty t^n e^{-a t}\, dt = \dfrac{n!}{(\lambda + \lambda_D)^{n+1}}\quad \text{with}\quad  a = \lambda + \lambda_D.
\end{align*}
\noindent Substituting back yields the closed-form expression
\[
\Pr[\widetilde{N} = n]
=
\lambda_D
\dfrac{\lambda^n}{(\lambda + \lambda_D)^{n+1}},
\qquad n = 0,1,2,\dots
\]

\noindent This expression can be rewritten as
\[
\Pr[\widetilde{N} = n]
=
\dfrac{\lambda_D}{\lambda + \lambda_D}
\left(
\dfrac{\lambda}{\lambda + \lambda_D}
\right)^n.\quad \text{Defining}
\]

\begin{align*}
    & p = \dfrac{\lambda_D}{\lambda + \lambda_D}, \quad \text{i.e.,} \quad 1 - p = 1 - \dfrac{\lambda_D}{\lambda + \lambda_D} = \dfrac{\lambda}{\lambda + \lambda_D}\\
    & \text{then} \quad f_{\widetilde{N}}(n) = \Pr[\widetilde{N} = n] = p(1-p)^n
\end{align*}

\noindent Therefore, $N$ follows a geometric distribution, counting the number of attacker steps before the defender returns. The probability that the attacker completes at least $n$ steps before the defender becomes active again is therefore
\[
\Pr[\widetilde{N} \geqslant n] = \sum_{k = n}^{\infty} p(1-p)^k = (1-p)^n = \left( \dfrac{\lambda}{\lambda + \lambda_D} \right)^n.
\]

\section{Case studies and numerical illustration} \label{cases_and_results}

For our experiments, We analyze three cases: two robotic systems and a virtual network deployed by Dynatrace company. For each, we present computational results and discuss their practical implications. To align with the literature while clarifying differences, we employ modified versions of the original attack graphs. In these graphs, some internal nodes, originally classified as targets, have descending nodes that are also targets. To maintain consistency with our model, where target nodes the null external degree nodes, we treat original target nodes with outgoing edges as intermediate nodes. These intermediates represent potential defensive points whose protection can prevent subsequent attacks. Consequently, the optimal defense often involves safeguarding selected sub-targets. Additionally, the nodes included in the defense policy depend on the defender’s assumed knowledge of the network, as some nodes may be inaccessible due to limited awareness.

\subsection{Case studies}
\noindent The different attack graphs evaluated are then described as follows.
\begin{enumerate}
    \item \emph{Modular articulated robotic arm (MARA):} Designed for professional industrial applications such as pick-and-place tasks, the MARA modular robot represents a significant advancement in collaborative technology. Its architecture is built around individual sensors and actuators that each possess native ROS 2.0 integration, allowing for a seamless physical expansion of the system. This decentralized architecture ensures precise synchronization, deterministic communication, and a structured hardware-software lifecycle. As reported by Alias Robotics (2019)\footnote{https://aliasrobotics.com/case-study-threat-model-mara.php}, MARA combines modular flexibility with robust standardized frameworks, expanding operational capabilities.

    \item \emph{Mobile industrial robotics (MiR100):} Studies by Alias Robotics (2020, 2021) identify key cybersecurity vulnerabilities in the MiR100, an autonomous mobile robot used to optimize internal logistics. In controlled experiments, a single compromised unit demonstrated how attackers could move laterally across subsystems, ultimately breaching the safety system. This exposes risks to the robot’s secure, automated operation despite its intended workflow optimization~\footnote{https://aliasrobotics.com/case-study-threat-model-mara.php, https://aliasrobotics.com/case-study-pentesting-MiR.php}.

    \item \emph{Unguard Virtual Lab:} The Unguard platform is a cloud-native microservices demo resembling a social media application, with a three-tier architecture (frontend, API, backend) comprising eight services, a load generator, and databases such as Redis and MariaDB. It is intentionally designed with critical vulnerabilities, including remote code execution (RCE), server-side request forgery (SSRF), and SQL injection, exposing internet-facing proxies like Envoy to potential compromise \cite{Dynatrace_LLC_Unguard_An_Insecure}. To analyze these threats, we employed the Cybersecurity AI (CAI) tool, which automates reconnaissance and exploitation through autonomous security agents \cite{mayoral2025cai, mayoral2025hacking}. Different agents yield different attack graphs depending on their capabilities; here, we selected the "Bug Bounter" agent for multi-step attack planning and service version detection. This produces a comprehensive attack graph that visualizes exploitation paths and lateral movements, enabling proactive identification and mitigation of systemic risks.
\end{enumerate}

\noindent For each attack graph above, we compare alternative defense strategies to assess their performance and practical effectiveness.

\subsection{Defender's strategies comparison}

To assess the performance of the proposed optimal strategy, we compare it with two intuitive heuristics that a defender might employ. Specifically, we evaluate the attacker’s success probability against a rational defender using the optimal strategy versus an uninformed defender following one of the heuristics defined below.

\begin{itemize}
    \item \textbf{Shortest-path heuristic:} When the defender lacks knowledge of which attack paths the attacker may choose, a natural approach is to assume the attacker follows a shortest path to the target node. Formally, given an attack graph $G = (V, E)$ with target $v_t \in V$, the defender assumes the attacker selects a path $\pi^\ast \in \argmin_{\pi: v \leadsto v_t} |\pi|$, where $|\pi|$ is the number of edges along $\pi$. This heuristic provides a baseline, reflecting the defender’s assumption that the attacker prefers the simplest route toward $v_t$.

    \item \textbf{Random defense heuristic:} In most adversity scenarios like ours, the defender knows that the attacker will not necessarily follow the shortest path as described previously, for some reason (e.g., lack of computational resources or intelligence), the defender cannot adopt the recommended optimal strategy.  We then study a defense strategy of the defender that consists of deploying resources randomly in the network, depending on its set of actions. This approach ignores specific attack patterns and serves as a benchmark for uninformed defensive actions.
\end{itemize}
 
We simulated all three scenarios (Stackelberg, Blind, and Dirichlet frameworks) on each attack graph, with the defender acting at random intervals ($\lambda_{\text{defender}} = 1$) while the attacker performs an average of two moves per time unit ($\lambda_{\text{attacker}} = 2$). Moreover, it should be remembered that the defender in his optimal strategy has no information on the attacker’s initial positions and therefore assumes a uniform distribution over all the nodes to which he has access to act.

\subsection{Results of experiments}

\subsubsection{MARA case}

The MARA robot’s attack graph comprises $9$ nodes and $9$ edges, including a single entry point (node $1$) and two targets (nodes $6$ and $9$), as shown in Figure~\ref{MARA_graph}. Figure~\ref{MARA_result_summary} summarizes the performance of the optimal strategy against the shortest-path and random heuristics. Node prioritization, represented by circle sizes on the graph (size corresponds to the priority of the node in the protection order), highlights that the defender focuses on nodes closest to the targets, reflecting the practical consideration that attackers may be stealthy and can initiate attacks from multiple points. Quantitatively, the optimal strategy consistently outperforms the heuristics: the attacker’s probability of success is significantly higher under the shortest-path or random defense than when the defender follows the proposed optimal strategy. This confirms the advantage of strategically allocating resources even under uncertainty about the attacker’s initial position. Notably, in the Dirichlet scenario, analytical results indicate that a defender using the Stackelberg-optimal strategy may be less effective if both players’ positions are drawn from the same Dirichlet distribution, highlighting the sensitivity of defense performance to the underlying assumptions about the attacker’s behavior.

\begin{figure}[ht!]
\vspace{-1cm}
\hspace*{-10cm} \includegraphics[width=3.5\linewidth]{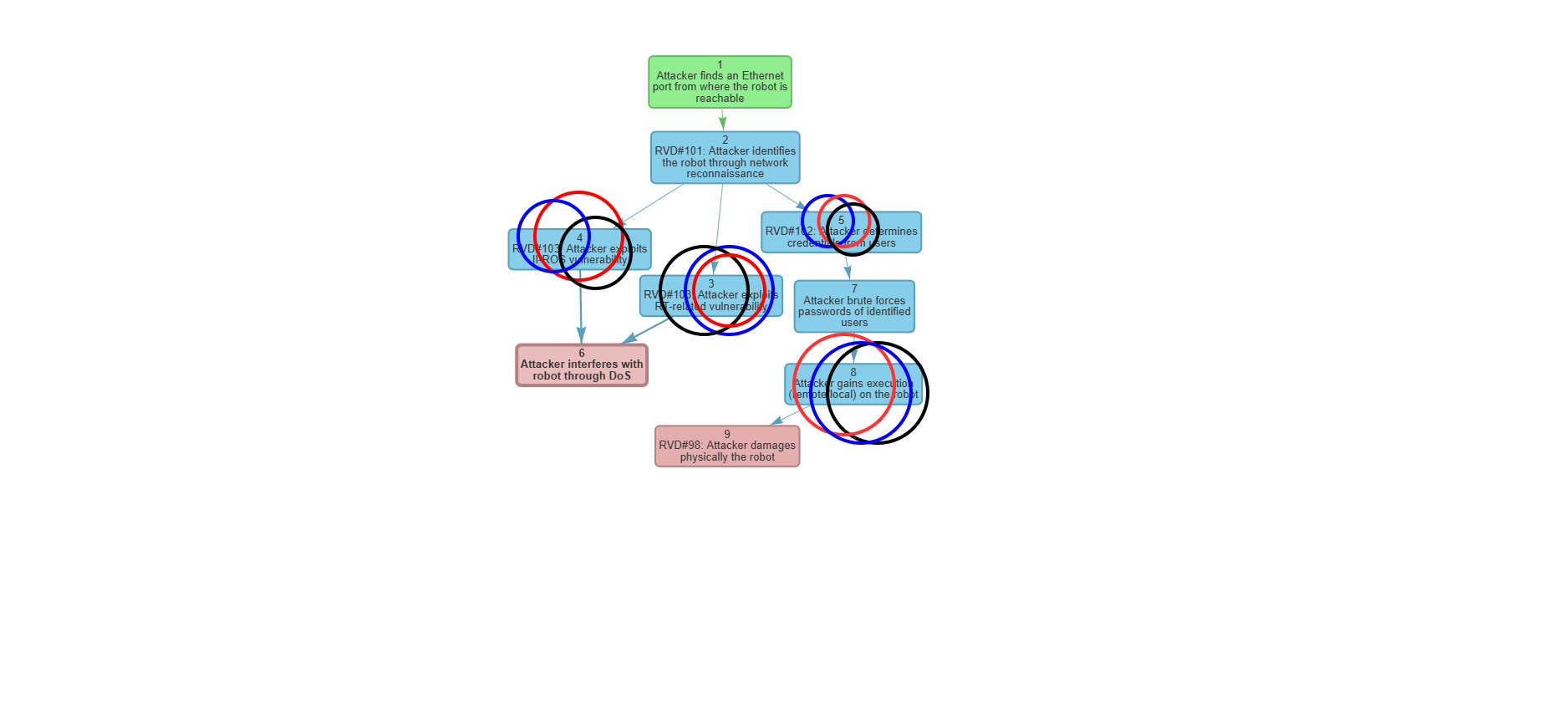}
\vspace{-5cm}
\caption{\textit{MARA network topology.} Visualization of the defender’s optimal node allocation across the three analyzed frameworks. Red circles indicate the Stackelberg-optimal deployment, blue circles correspond to the blind strategy, and black circles represent the Dirichlet-based strategy.}
\label{MARA_graph}
\end{figure}

\begin{figure*}[t]
    \centering
    \includegraphics[width=1\linewidth]{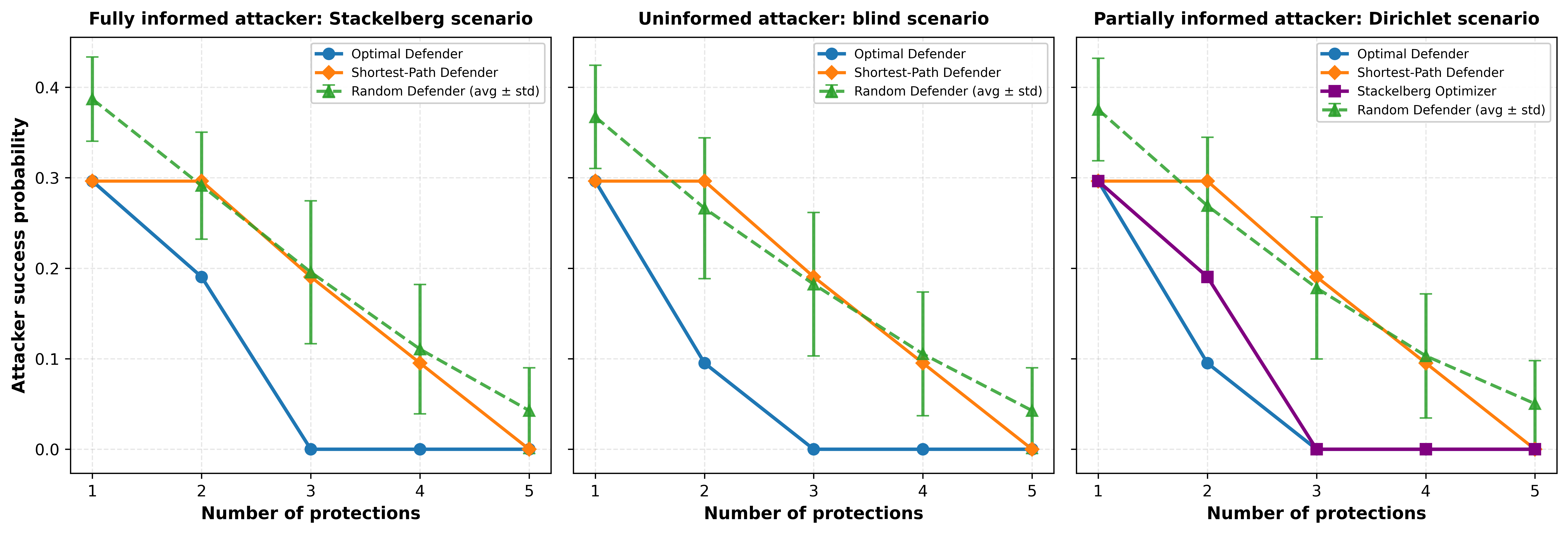}
    \caption{\textbf{Comparison of defensive strategies under varying attacker information: MARA case.} The results show that the optimal defense consistently achieves the lowest attacker success probability by prioritizing nodes closest to the targets. Both shortest-path and random heuristics yield significantly higher breach probabilities, confirming the benefit of strategic allocation under uncertainty. However, when both players follow a Dirichlet belief structure, applying the Stackelberg policy directly may reduce defensive performance, illustrating the sensitivity of outcomes to modeling assumptions. 
}
    \label{MARA_result_summary}
\end{figure*}

\begin{figure*}[t]
    \centering
    \includegraphics[width=1\linewidth]{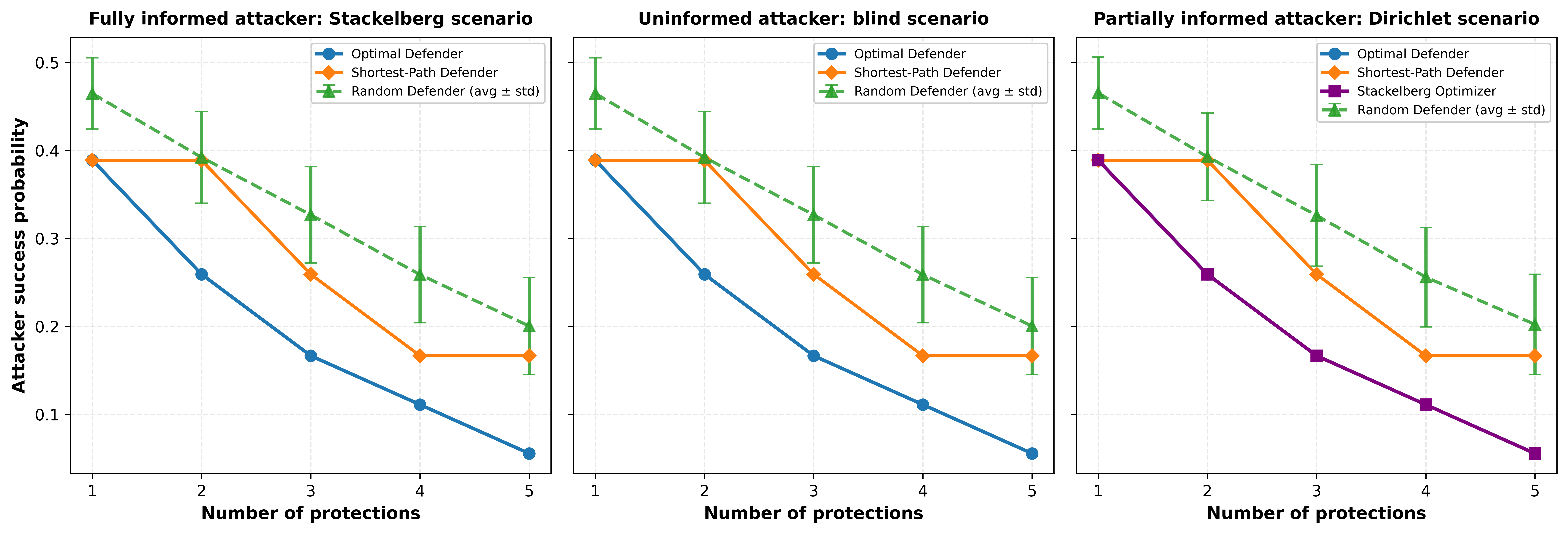}
    \caption{\textbf{Comparison of defensive strategies under varying attacker information: MiR100 case.} Due to low path diversity and the presence of dominant bottlenecks (notably node 15), all three game-theoretic frameworks converge to the identical performance curves. Topology dominates belief assumptions: once key points are protected, attacker behavior becomes effectively unique. This case illustrates that in highly constrained graphs, identifying structural bottlenecks is sufficient to recover the full benefit of strategic optimization.}
    \label{MiR100_result_summary}
\end{figure*}

\subsubsection{MiR100 case}

Figure~\ref{MiR100_graph} presents the MiR100 attack graph, composed of 16 nodes and 24 edges, with four entry points (nodes $1$–$4$) and four critical targets (nodes $12$, $13$, $14$, $16$). To make the prioritization clear, node priorities are visualized through bubble sizes, reflecting the allocation intensity of defensive resources. Figure~\ref{MiR100_result_summary} compares the optimal defense with the shortest-path and random heuristics.

\noindent \textbf{Structural dominance and bottlenecks.} A key feature of this graph is its low path diversity: many entry-target pairs are connected by a single feasible route. For example, entry node 1 can reach targets 12 and 13 only through $1 \to 5 \to 15 \to {12,13}$, and entry node 4 reaches target 14 only via $4 \to 7 \to 14$. Node 15 lies on all paths from entry 1 to targets ${12,13,16}$, while nodes ${7 ,8, 10, 15}$ collectively intercept most viable attack routes. Betweenness centrality confirms this structural concentration (BC$(15)=0.402$, BC$(10)=0.201$, BC$(8)=0.271$, BC$(7)=0.156$), clearly identifying node 15 as the principal bottleneck. Furthermore, with detection parameters $\lambda_{\text{attacker}} = 2$ and $\lambda_{\text{defender}} = 1$, the geometric success probability $P_{\text{success}}(\ell) \;=\; \left(\dfrac{\lambda_{\text{attacker}}}{\lambda_{\text{attacker}}  + \lambda_{\text{defender}}}\right)^{\!\ell} \;=\; \left(\dfrac{2}{3}\right)^{\ell}$,
decays rapidly with path length: a five-hop path succeeds with only $13.2\%$ probability compared to $44.4\%$ for a two-hop path. This geometric decay effectively removes longer alternatives from the attacker’s rational choice set, reinforcing the dominance of short, bottlenecked routes.

\noindent \textbf{Convergence of all three frameworks:} A central result is that the Stackelberg, Blind, and Dirichlet formulations yield identical optimal defender strategies and identical performance curves. This convergence probably stems from the graph topology. Because each state admits essentially a unique value-maximizing action for the attacker (i.e., $|\displaystyle\argmax_a Q(s, a)| = 1$ in most states), equilibrium, best-response, and belief-averaged policies coincide. Once the attacker’s behavior is invariant across frameworks, the defender faces the same optimization problem in all cases. Hence, topology dominates belief specification and information structure. The resulting optimal deployment ${15,8,7,9,11}$ consistently concentrates protection near targets and structural choke points, which is also visually reflected in the bubble prioritization.

\begin{figure}[ht!]
\vspace{-0.75cm}
    \hspace*{-7.95cm} \includegraphics[width=3\linewidth]{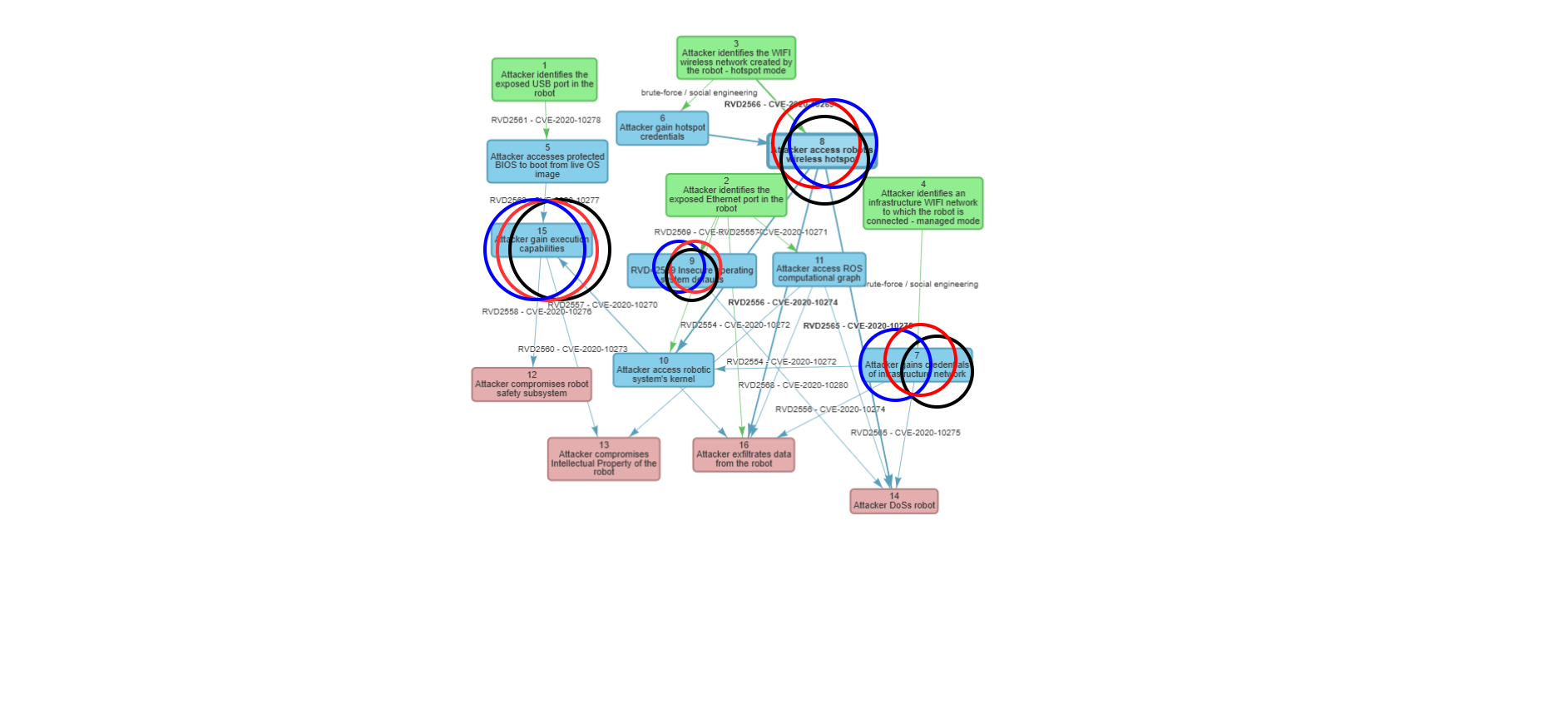}
    \vspace{-3.5cm}
    \caption{\textit{MiR100 network topology.} Illustration of how the defender prioritizes nodes under the three studied game-theoretic frameworks. Red circles show Stackelberg-optimal choices, blue circles indicate the blind approach, and black circles mark the Dirichlet-based allocation.}
    \label{MiR100_graph}
\end{figure}

\noindent \textbf{Heuristic performance gaps:} Although the optimal frameworks coincide, the heuristics reveal meaningful performance differences. The shortest-path heuristic matches the optimal strategy when $h = 1$ (both select node 15), since its dominance is evident even without strategic modeling. However, for $h \geqslant 2$, divergence appears, and this discrepancy increases with $h$, reaching attacker success probabilities of $0.167$ under the heuristic versus $0.025$ under the optimal strategy at $h=5$. The random defense performs substantially worse across all resource levels, with average attacker success rates approximately twice those of the optimal deployment and high variance across simulations. This variability illustrates the absence of systematic bottleneck coverage under random allocation. Overall, the MiR100 case represents a regime where structural constraints outweigh informational assumptions: low path diversity, clear choke points, and strong geometric decay neutralize differences between game-theoretic formulations. For practitioners, this implies that identifying and securing dominant bottlenecks is sufficient to recover the full strategic benefit of the model, independently of how attacker beliefs are specified.

\subsubsection{Unguard network case}

The Unguard attack graph (Figure~\ref{Unguard_graph}) is significantly more complex, comprising $33$ nodes and $57$ edges. The threat model involves a single entry node (node $1$, external attacker) and five critical target nodes $\{29,30,31,32,33\}$, corresponding respectively to full system compromise, sensitive data exposure, account takeover, unauthorized actions, and bypassed network segmentation. Figure~\ref{Unguard_result_summary} evaluates the optimal defense against heuristic strategies. As in previous cases, node importance is visualized through circle scaling.

\begin{figure}[ht!]
\vspace{-0.5cm}
    \hspace*{-7.75cm} \includegraphics[width=2.75\linewidth]{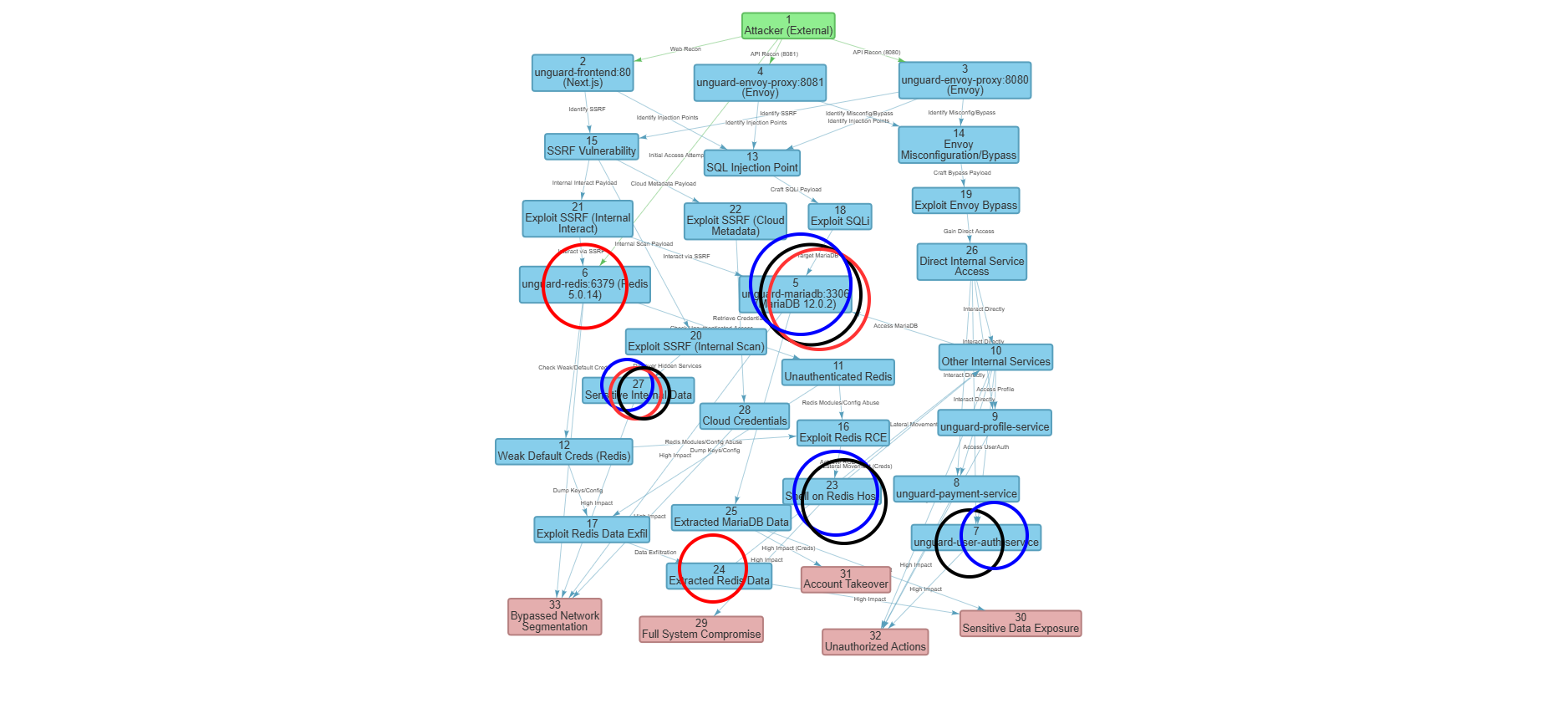}
    \vspace{-1cm}
    \caption{\textit{Unguard network topology.} Representation of the defender’s optimal protections across the three cases considered. The red markers denote Stackelberg-optimal placement, the blue markers the blind framework, and the black markers correspond to the Dirichlet scenario.}
    \label{Unguard_graph}
\end{figure}

\begin{figure*}[t]
    \centering
    \includegraphics[width=1\linewidth]{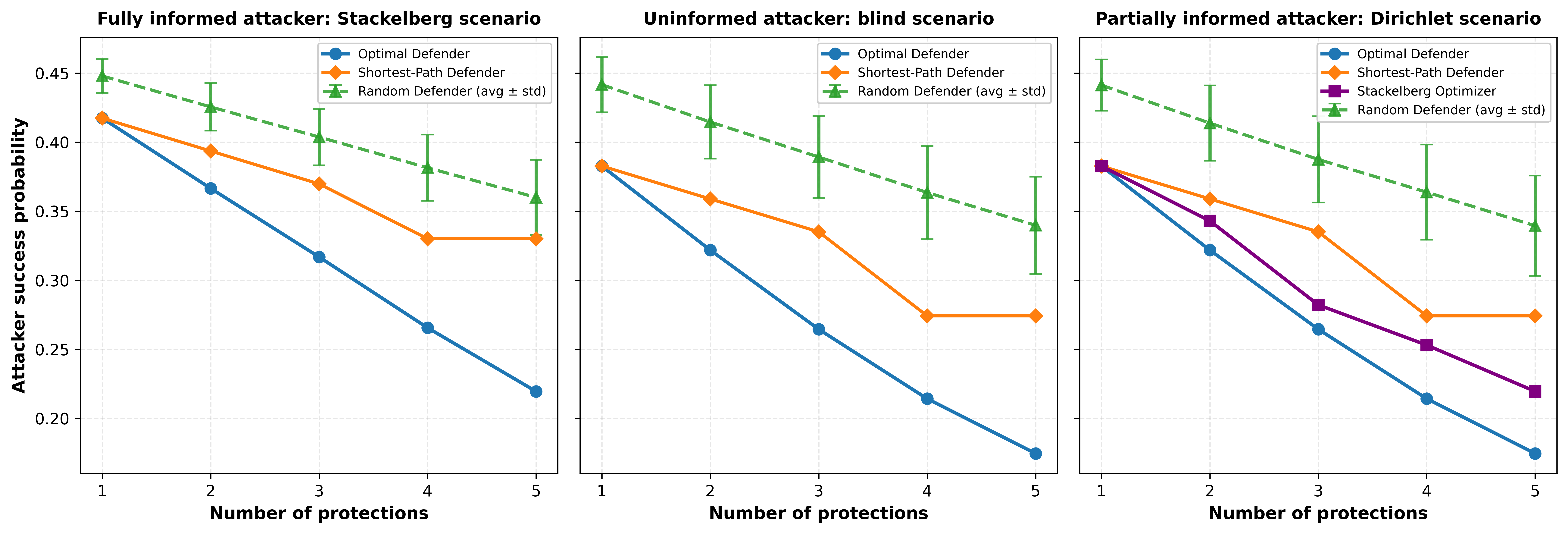}
    \caption{\textbf{Comparison of defensive strategies under varying attacker information: Unguard case.} In this high-diversity network, no single bottleneck exists, and multiple redundant attack paths make the attacker’s optimal choice belief-dependent. The three game-theoretic frameworks produce distinct optimal strategies, with the highest impact achieved by prioritizing nodes that appear across multiple attack vectors (e.g., node 5). The shortest-path heuristic performs poorly, while random allocation is inconsistent. Strategic optimization reduces attacker success from $0.275$ (shortest-path) to $0.09$ at $h = 5$, highlighting a $3\times$ improvement and demonstrating the critical value of game-theoretic reasoning in complex, redundant networks.}
    \label{Unguard_result_summary}
\end{figure*}

Unlike the MiR100 graph, Unguard exhibits high path diversity structured around four largely independent attack vectors: Redis, SQL injection, Envoy bypass, and SSRF. A key structural property is that no single node intercepts all attack routes. Hence, there is no main bottleneck. Each vector also contains redundant sub-paths, preventing complete neutralization of a vector through a single protection. As a result, the three game-theoretic frameworks produce distinct optimal strategies and performance curves. The divergence arises because the attacker’s optimal action becomes belief-dependent when multiple viable routes coexist. Moreover, the optimal strategies achieve the lowest attacker success probabilities, decreasing sharply from approximately $0.375$ at $h = 1$ to $0.09$ at $h = 5$. The reduction is driven by prioritizing nodes that simultaneously disrupt multiple attack chains. For example, node 5 (MariaDB) lies on both the SQLi path ($13 \to 18 \to 5$) and the SSRF path ($21 \to 5$); protecting it disables two independent vectors. Subsequent protections similarly target high-impact shared components, producing consistent gains. This confirms that, under uncertainty about attacker progression, concentrating defenses near shared downstream resources yields robust protection. Importantly, consistent with the analytical results, directly applying the Stackelberg solution in the Dirichlet setting leads to performance degradation. The Stackelberg strategy is optimized for a fixed belief, whereas the Dirichlet framework requires robustness across a distribution of possible attacker beliefs.

\noindent \textbf{Heuristic limitations.}  The shortest-path heuristic exhibits nearly flat performance, reflecting its structural inadequacy in redundant graphs. Blocking the shortest path within a vector simply redirects the attacker to a parallel alternative. Furthermore, while random allocation occasionally covers shared components by chance, it offers no reliability guarantee. The advantage of strategic reasoning, then, is most pronounced in this complex topology. At $h = 5$, the optimal defender achieves an attacker success probability of $0.09$, compared to $0.275$ under the shortest-path heuristic, corresponding to a $3\times$ increase in attacker success when game-theoretic optimization is ignored. This substantial gap quantifies the cost of heuristic reasoning in networks characterized by redundancy and multiple independent attack vectors.

\section{Conclusion}\label{conclusion}

This work investigates optimal defensive strategies to impede a strategically behaving attacker already present within a network. Recognizing that interaction outcomes depend on the attacker's information quality, we studied three scenarios: a fully informed attacker, a non-informed attacker, and a partially informed attacker. Each scenario was evaluated on real attack graphs, including the MARA robot, MiR100 robot, and the Unguard microservices network. Our results reveal a clear distinction: in low-diversity graphs with critical bottlenecks, all frameworks converge to the same optimal strategy, indicating that network topology can dominate information assumptions; in high-diversity networks with redundant paths, optimal strategies diverge, and the choice of framework significantly influences performance, with game-theoretic defenses reducing, in some instances, attacker success by more than a factor of three compared to heuristic approaches. These findings provide practical guidance: assess network structure, identify bottlenecks nodes, and select the strategic framework according to attacker flexibility and uncertainty, by prioritizing the nodes close to the targets. Overall, effective defense depends on both graph topology and attacker modeling, and rigorous, topology-aware, game-theoretic planning offers substantial and quantifiable advantages over conventional heuristics. Nevertheless, the proposed framework presents several limitations. The analysis assumes that the defender operates without observation feedback and therefore does not update beliefs regarding the system’s infection state or the attacker’s location within the network. Furthermore, detection mechanisms are modeled under the assumption of perfect reliability, with sensors achieving a unit probability of detection, which may not accurately reflect real-world operational environments.

\end{document}